\documentclass[a4paper,onecolumn,11pt,accepted=2023-09-27]{quantumarticle}
\pdfoutput=1
\usepackage[utf8]{inputenc}
\usepackage[english]{babel}
\usepackage[T1]{fontenc}
\usepackage{amsmath}

\usepackage{tikz}
\usepackage{lipsum}

\usepackage{tabularx,lipsum,environ,amsmath,amssymb}
\usepackage{algorithm}
\usepackage{algorithmicx}
\usepackage{amsmath}
\usepackage{amsthm}
\usepackage[noend]{algpseudocode}
\usepackage{bbm}
\usepackage{bm}
\usepackage{booktabs}
\usepackage{color}
\usepackage{hyperref}
\usepackage{dcolumn}
\usepackage{graphicx}
\usepackage{lipsum}
\usepackage{listings}
\usepackage{makecell}
\usepackage{physics}
\usepackage{qcircuit}
\usepackage{ragged2e}
\usepackage{subfigure}
\usepackage{tabularray}
\usepackage[numbers,compress]{natbib}

\makeatletter
\newcommand{\problemtitle}[1]{\gdef\@problemtitle{#1}}
\newcommand{\probleminput}[1]{\gdef\@probleminput{#1}}
\newcommand{\problemquestion}[1]{\gdef\@problemquestion{#1}}
\NewEnviron{problem}{
  \problemtitle{}\probleminput{}\problemquestion{}
  \BODY
  \par\addvspace{.5\baselineskip}
  \noindent
  \begin{tabularx}{\textwidth}{@{\hspace{\parindent}} l X c}
    \multicolumn{2}{@{\hspace{\parindent}}l}{\@problemtitle} \\
    \textbf{Input:} & \@probleminput \\
    \textbf{Question:} & \@problemquestion
  \end{tabularx}
  \par\addvspace{.5\baselineskip}
}
\makeatother

\newtheorem{thm}{Theorem}
\newtheorem{coro}[thm]{Corollary}
\newtheorem{lem}[thm]{Lemma}
\newtheorem{definition}{Definition}

\AtBeginDocument{%
  \heavyrulewidth=.08em
  \lightrulewidth=.05em
  \cmidrulewidth=.03em
  \belowrulesep=.65ex
  \belowbottomsep=0pt
  \aboverulesep=.4ex
  \abovetopsep=0pt
  \cmidrulesep=\doublerulesep
  \cmidrulekern=.5em
  \defaultaddspace=.5em
}
\begin{document}

\title{Quantum Encoding and Analysis on Continuous Time Stochastic Process with Financial Applications}

\author{Xi-Ning Zhuang}
\affiliation{CAS Key Laboratory of Quantum Information,  University of Science and Technology of China,  Hefei,  230026,  China}
\affiliation{Origin Quantum Computing,  Hefei,  China}
\orcid{0000-0001-5118-5066}
\author{Zhao-Yun Chen}
\affiliation{Institute of Artificial Intelligence,  Hefei Comprehensive National Science Center}
\orcid{0000-0002-5181-160x}
\author{Cheng Xue}
\affiliation{Institute of Artificial Intelligence,  Hefei Comprehensive National Science Center}
\orcid{0000-0003-2207-9998}
\author{Yu-Chun Wu}
\email{wuyuchun@ustc.edu.cn}
\affiliation{CAS Key Laboratory of Quantum Information,  University of Science and Technology of China,  Hefei,  230026,  China}
\affiliation{CAS Center for Excellence and Synergistic Innovation Center in Quantum Information and Quantum Physics,  University of Science and Technology of China,  Hefei, 230026,  China}
\affiliation{Hefei National Laboratory, University of Science and Technology of China, Hefei 230088, China}
\affiliation{Institute of Artificial Intelligence,  Hefei Comprehensive National Science Center}
\orcid{0000-0002-8997-3030}
\author{Guo-Ping Guo}
\email{gpguo@ustc.edu.cn}
\affiliation{CAS Key Laboratory of Quantum Information,  University of Science and Technology of China,  Hefei,  230026,  China}
\affiliation{CAS Center for Excellence and Synergistic Innovation Center in Quantum Information and Quantum Physics,  University of Science and Technology of China,  Hefei,  230026,  China}
\affiliation{Hefei National Laboratory, University of Science and Technology of China, Hefei 230088, China}
\affiliation{Institute of Artificial Intelligence,  Hefei Comprehensive National Science Center}
\affiliation{Origin Quantum Computing,  Hefei,  China}
\orcid{0000-0002-2179-9507}
\maketitle

\begin{abstract}
Modeling stochastic phenomena in continuous time is an essential yet challenging problem. Analytic solutions are often unavailable, and numerical methods can be prohibitively time-consuming and computationally expensive. To address this issue, we propose an algorithmic framework tailored for quantum continuous time stochastic processes. This framework consists of two key procedures: data preparation and information extraction.
The data preparation procedure is specifically designed to encode and compress information, resulting in a significant reduction in both space and time complexities. This reduction is exponential with respect to a crucial feature parameter of the stochastic process. Additionally, it can serve as a submodule for other quantum algorithms, mitigating the common data input bottleneck. 
The information extraction procedure is designed to decode and process compressed information with quadratic acceleration, extending the quantum-enhanced Monte Carlo method. 
The framework demonstrates versatility and flexibility, finding applications in statistics, physics, time series analysis and finance. Illustrative examples include option pricing in the Merton Jump Diffusion Model and ruin probability computing in the Collective Risk Model, showcasing the framework's ability to capture extreme market events and incorporate history-dependent information. 
Overall, this quantum algorithmic framework provides a powerful tool for accurate analysis and enhanced understanding of stochastic phenomena.
\end{abstract}

\section{INTRODUCTION}

The continuous time stochastic process (CTSP) is a fundamental mathematical tool that encompasses various well-known stochastic processes, such as \textit{Poisson point process},  \textit{compound Poisson process},  \textit{L\'{e}vy process}, \textit{continuous Markov process}, and doubly stochastic \textit{Cox process}. 
CTSP plays a crucial role in modeling stochastic phenomena occurring in continuous time variables, and its applications span multiple disciplines including finance, physics, statistics, and biology \cite{papapantoleon2008introduction, barndorff2001levy, liggett2010continuous, anderson2012continuous, dassios2003pricing}. 
However, compared to its discrete counterpart, CTSP is considered more complex due to its continuous path. 
Analytic solutions or explicit formulas for the underlying stochastic differential equations and quantities are often unavailable for practical problems. 
Moreover, the application of numerical methods like Monte Carlo simulation can be computationally demanding, requiring unexpectedly large storage and computation resources. 
This is primarily due to the exponential growth of the sampling space as time slices become finer, and the intricate nature of the evolution dynamics within the process \cite{ross1996stochastic,kozachenko2016simulation}. 
These complexities make efficient and accurate analysis of CTSPs a challenging task in various fields of study. 

The developments of quantum processors \cite{feynman1985quantum, divincenzo2000physical, arute2019quantum} and algorithms \cite{orus2019quantum,egger2020quantum,herman2022survey,wilkens2023quantum,mcardle2020quantum,outeiral2021prospects,emani2021quantum,ma2020quantum,cao2018potential} have revealed that quantum computation has great potential beyond classical computers.
Thus it would be one powerful tool to solve the above challenging problems. However, there are still two fundamental challenges to be overcome when implementing the storage and analysis of CTSP with a quantum computer. 

The first challenge is the data-loading procedure of CTSP: As a common and primary bottleneck faced by quantum machine learning and many other algorithms\cite{schuld2018supervised}, the state preparation problem has been studied and discussed in many works \cite{grover2002creating, vazquez2021efficient, rattew2022preparing}. 
Quantum analog simulators are proposed to predict the future of a specified class of CTSP named renewal process with less past information, and an unbounded memory requirement reduction is made in \cite{elliott2018superior,elliott2019memory,elliott2021quantum}. The quantum advantage in simulating stochastic processes is further discussed in \cite{korzekwa2021quantum}. 
Nevertheless, those works have not totally solved the CTSP preparation problem. 
The information is not digital-encoded, so the analysis, such as quantum-enhanced Monte Carlo method is not easy to implement.
Moreover, the type of CTSP is restricted to the renewal process. Furthermore, the storage of the entire path is absent, leading to crucial practical problems for financial engineering, quantitative trading, and many other path-dependent scenarios. 

The second challenging problem is the information extraction of CTSP. 
Quantum-enhanced Monte Carlo method, which involves generating random paths and computing the expectation of desired values, has been shown to achieve quadratic quantum speedup in computing the final value of discrete sampling paths \cite{montanaro2015quantum, rebentrost2018quantum, stamatopoulos2020option, martin2019towards, woerner2019quantum, blank2021quantum}. 
Additionally, a weightless summation of discrete paths, targeting Asian-type option pricing, is proposed in the work by Stamatopoulos et al. \cite{stamatopoulos2020option}. 
However, the quantum-enhanced Monte Carlo method for CTSP, which involves computing the expectation of a weighted integral over the CTSP paths and extracting history-sensitive information (such as the first-hitting time problem), remains unresolved.

In this work, we establish a framework to solve the problems of state preparation and information extraction of quantum continuous time stochastic process (QCTSP). 
Two representations encoding general QCTSP are introduced.
The corresponding state preparation method is developed to prepare the QCTSP with less qubit number requirement,  higher flexibility, and more sensitivity to discontinuous jumps that are important to model extreme market events such as \textit{flash crash}.
An observation of the CTSP holding time is made, inducing further reductions in the circuit depth and the gate complexity of QCTSP. 
Specific quantum circuits are designed for most of the often-used CTSPs with exponential reduction of both qubit number and circuit depth on the key parameter of QCTSP named holding time $\tau_{avg}$ (as summarized in table \ref{tab:tab2}). 
As for the information extraction problem of QCTSP, the weightless integral and the arbitrary time-weighted integral of the QCTSP can be efficiently evaluated by computing the summation of directed areas. 
Furthermore, this method enables the quantum-enhanced Monte Carlo framework to extend to the continuous-time regime, admitting a quadratic quantum speed up. Moreover, by introducing a sequence of flag qubits, our method can extract the history-sensitive information that is essential and indispensable for quantitative finance, path-dependent option pricing, and actuarial science, to cite but a few examples.

Applications of computing the European type option price in the \textit{Merton Jump Diffusion Model} and the ruin probability in the \textit{Collective Risk Model} are given.
 The first application of evaluating a European-type option has been studied in previous work \cite{rebentrost2018quantum, stamatopoulos2020option, martin2019towards, blank2021quantum}, while our method takes the more practical situation of the discontinuous price movement into consideration, and the simulation result is consistent with the \textit{Merton} formula. 
 The second application of computing the ruin probability opens up new opportunities in the area of  insurance as being the first time that quantum computing has been introduced into insurance mathematics up to known, illustrating the great potential power of QCTSP.

In addition to the theoretical importance and rich applicability of QCTSP mentioned above, our work also benefits from its low requirement of input data and circuit connectivity. 
It can be utilized as an input without quantum random access memory (qRAM) \cite{giovannetti2008quantum,  giovannetti2008architectures,  hong2012robust}, partly mitigating the input problem that quantum machine learning and many other quantum algorithms are confronted with. 
Simultaneously, an indicator is given in this work to characterize not only the preparation procedure's complexity but the circuit connectivity as well.

The structure of this article is as follows: Given the mathematical notations for readability and the formal problem statement, the main framework of QCTSP is sketched in Section \ref{sec:1}. 
The main results of state preparation and information extraction are organized in Section \ref{sec:prepare} and Section \ref{sec:extract}, respectively. 
Followed by two applications of option pricing and insurance mathematics in Section \ref{sec:4}, the discussion, together with an introduction of future work, is given in Section \ref{sec:5}. 
The detailed construction of specific modified subcircuits and proofs can be found in the Appendix.

\section{QUANTUM CONTINUOUS TIME STOCHASTIC PROCESS}\label{sec:1}

The mathematical notation is given in the first subsection, followed by the formal definition of CTSP and the two representations of QCTSP. Then the core problems of state preparation and information extraction of QCTSP are stated. A brief framework of our QCTSP is described at the end of this section.

\subsection{Preliminaries}
\begin{table}\centering
\caption[c]{Mathematical Symbols} \label{tab:tab1}
\begin{tabular}{cl}
\toprule
Notation & Nomenclature\\
\midrule
{$\mathbb{P}$} & {The probability.}\\
{$p$} & {The probability amplitude.}\\
$\Omega$ & Space of CTSP.\\
$\bar{\Omega}_n$ & Space of CTSP of $n$ deterministic pieces.\\
$\mathcal{S}$ & Space of states of point.\\
$S$ & the size of the Space of states $S=\abs{\mathcal{S}}$.\\
$\mathbb{Z}_S$ & The set of index $\{1, ..., S\}$.\\
$T$ & Time slices.\\
$\mathbb{Z}_T$ & The set of time $\{1, ..., T\}$.\\
$n$ & Number of Pieces.\\
$X(t)$ & A CTSP.\\
$X_j$ & The $j^{th}$ piece's value variable.\\
$x_{j,k_j}$ & A realization of $X_j$, taking the $k_j\in\mathcal{S}$ index.\\
{$\bm{k}$} & \makecell[l]{{The index vector denoting the} \\{state space indexes of the path.}}\\
$Y_j$ & The $j^{th}$ piece's increment.\\
$\tau_j$ & The $j^{th}$ piece's holding time.\\
$t_j$ & A realization of $\tau_j$ in path.\\
{$\bm{t}$} &{The holding time vector of one path realization.}\\
$T_j$ & The $j^{th}$ piece's cumulative time.\\
$\mathcal{ML}(X)$ & The memory length of $X(t)$.\\
\bottomrule
\end{tabular}
\end{table}

{Intuitively, a continuous time stochastic process is a sequence of random variables $X(t)$ indexed by a non-negative continuous parameter $t$, such as the position of a particle or the price of a financial instrument.}
Wherein each random variable $X(t)$ lies in the same space named the state space $S$.
In this article we consider the state space $\mathcal{S}$ with a finite size $\abs{\mathcal{S}} = S \in\mathbb{N_+}$, whereas the more generic continuous case can be solved via a discretization.
For the current time $t = t_0$, the random variable $X(t_0)$ follows a distribution $\mathcal{F}(t_0)$ determined by the previous history $\{X(t):0\le t < t_0\}$, {i.e., 
if an observation is made, the random variable $X(t_0)$ falls into one specific state, as known as a realization, in the underlying space $\mathcal{S}$.}
In the generic case, the distribution $\mathcal{F}(t_0)$ also varies following the evolutionary dynamics of this CTSP.
More formally, given the mathematical symbols in Table \ref{tab:tab1}, a CTSP can be defined as 
\begin{definition}\label{def:1}
\textbf{(Continuous Time Stochastic Process)} Given a discrete space $\mathcal{S}$, a continuous time stochastic process $\{X(t):t\ge0\}$ is a stochastic process defined on the continuous variable $t$, where for each $t_0$, the state $X(t_0)\in \mathcal{S}$ is a stochastic variable  whose possibility distribution is determined by $\mathbb{P}[X(t_0)]=\mathbb{P}[X(t_0)|X(t):0\le t < t_0]$.
\end{definition}
\noindent For a given CTSP, the most significant distinction from the discrete case is the uncountable dimensional space of continuous sample paths denoted by $\Omega$. 
As the time slices become thinner, the sample space becomes a disaster for both theoretical analysis and simulation. 
Henceforth, the reduced space $\bar{\Omega}_n$ consists of those paths that can be divided into finite $n\in\mathbb{N}_+$ in-variate pieces is considered, where each $\{X(t):t\ge0\}$ in $\bar{\Omega}_n$ is a piece-wise determined random function $X(t)=X_j$ for $T_{j-1}\le t < T_j$ with $T_0=0$. 

\begin{figure}[h]
\centering
\includegraphics[width=\textwidth]{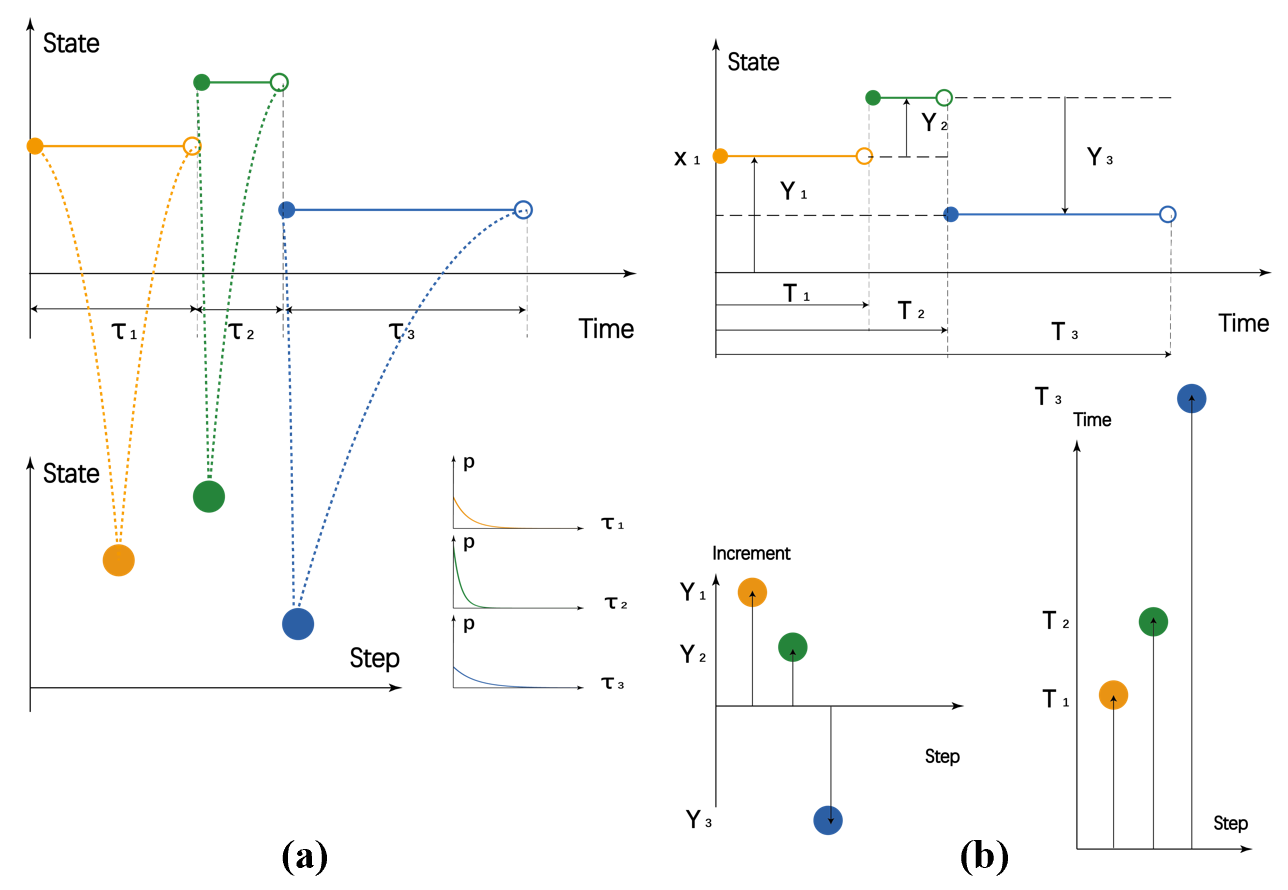}
\caption{\textbf{Embed a Continuous Time Stochastic Process to a pair of Discrete Stochastic Processes.} In this figure, a continuous time stochastic process is embedded into discrete stochastic processes via two different representations. \textbf{(a)}  As shown in the left subfigure, the $j^{th}$ in-variant piece of CTSP can be compressed into a discrete random variable $X_j$ (denoted by the point) in the state space, together with a holding time variable $\tau_j$ follows a given probability distribution (Illustrated by a picture of the probability density function). \textbf{(b)} As shown in the right subfigure, the $j^{th}$ in-variant piece of CTSP can be compressed into a discrete random variable of increment $Y_j=X_j-X_{j-1}$ (denoted by the vertical arrow) in the state space, together with an ending time variable $T_j$ (denoted by the horizontal arrow).}
\label{encoding}
\end{figure}

Thus the $j^{th}$ piece can be recorded as a pair of random variables $(X_j, \tau_j)$, where $X_j\in\mathcal{S}$ is a discrete random variable that denotes the $j^{th}$ piece's state, and $\tau_j=T_j-T_{j-1}\in\mathbb{N^+}$ is a strictly positive random variable that denotes lifetime length of the $j^{th}$ piece. 
In this way, the whole CTSP is encoded via a sequence of random variable pairs $\{(X_j,\tau_j):1\le j \le n\}$, which will be called the \emph{holding time representation} hereafter (as illustrated in Figure \ref{encoding} \textbf{(a)}). 
More precisely, we have:
\begin{definition}\label{def:2}
\textbf{(Holding Time Representation of Quantum Continuous Time Stochastic Process )} Given a continuous time stochastic process $\{X(t):t\ge0\}$, the holding time representation of quantum continuous time stochastic process is a sequence of stochastic pairs $(X_j, \tau_j)$ satisfying:
(1) $X(t)=X_j:T_{j-1}\le t < T_j$, 
(2) $\tau_j = T_j - T_{j-1}$, and 
(3) $\mathbb{P}[X_j, \tau_j]=\mathbb{P}[X_j, \tau_j|X_{j'},\tau_{j'}:1\le j'<j]$.
\end{definition}
An alternative method named \emph{increment representation} to encode a CTSP is to consider an equivalent sequence of random variable pairs $\{(Y_j,T_j):1\le j \le n\}$ where $Y_1 = X_1$ and $Y_j = X_j - X_{j-1}$ are the increments of the CTSP (as illustrated in Figure  \ref{encoding} \textbf{(b)}):
\begin{definition}\label{def:3}
\textbf{(Increment Representation of Quantum Continuous Time Stochastic Process )} Given a continuous time stochastic process $\{X(t):t\ge0\}$, the increment representation of quantum continuous time stochastic process is a sequence of stochastic pairs $(Y_j, T_j)$ satisfying:
(1) $X(t)=X_j:T_{j-1}\le t < T_j$, 
(2) $Y_j = X_j - X_{j-1}$, and 
(3) $\mathbb{P}[Y_j, T_j]=\mathbb{P}[Y_j, T_j|Y_{j'},T_{j'}:1\le j'<j]$.
\end{definition}
\noindent Instead of the uniform sampling method, these two encoding methods benefit from the reduction of qubit number (as shown in Figure  \ref{DSP_QCTSP_Circuit}), as well as the ability to capture randomly occurring discontinuous jumps (as circled in Figure \ref{levy_picture} \textbf{(a)} and Figure \ref{cox_picture} \textbf{(a)}). These two equivalent representations are proposed simultaneously to efficiently compute path-dependent and history-sensitive information, as stated later in Section \ref{sec:extract}.

\begin{figure}[h]
\centering
\includegraphics[width=0.85\textwidth]{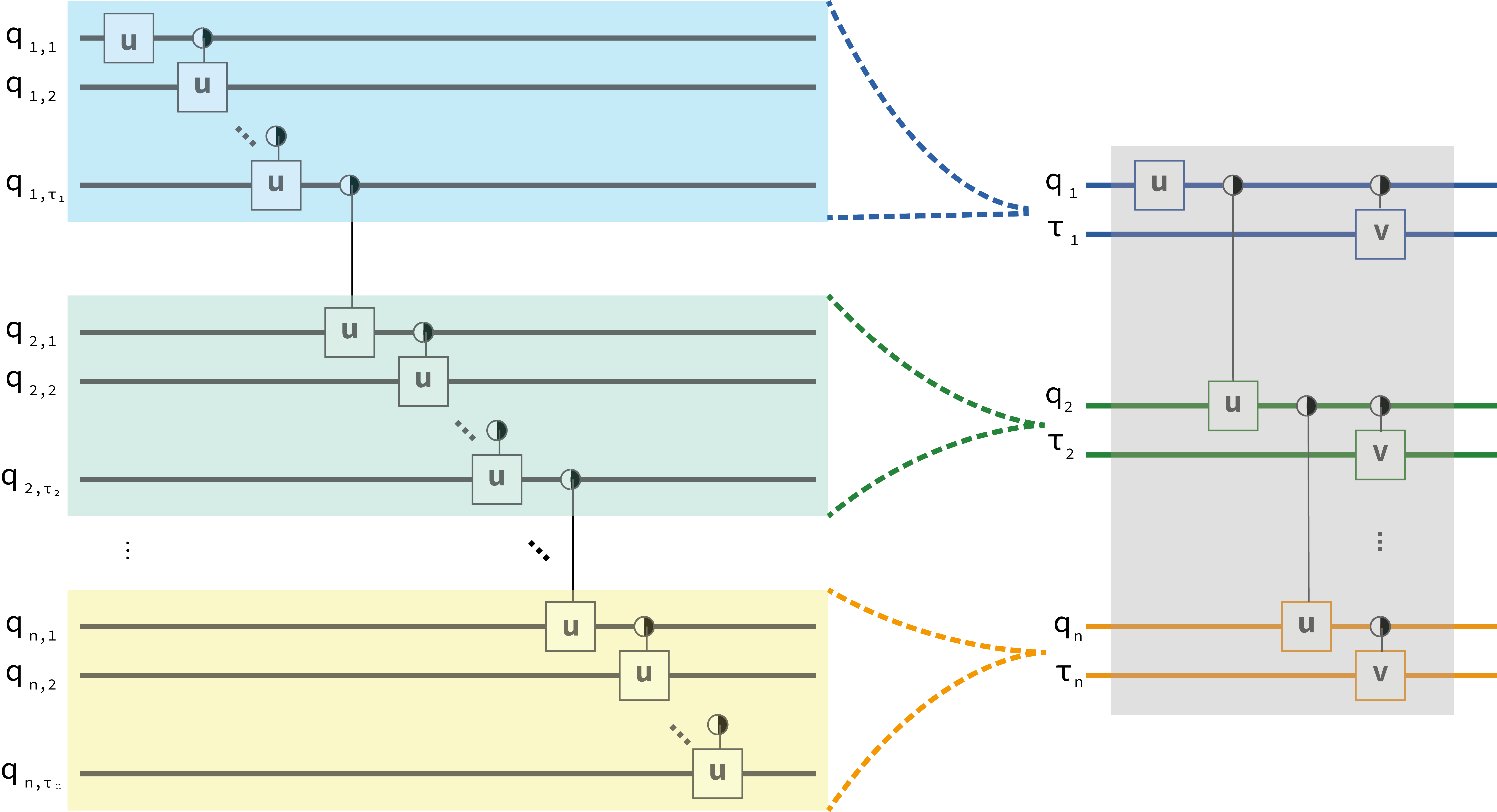}
\caption{\textbf{Embed a Continuous Time Stochastic Process to a pair of Discrete Stochastic Processes.}  In this figure, a continuous time stochastic process is embedded into a sequence of stochastic pairs. Both the qubit number and the circuit depth is reduced as shown.}
\label{DSP_QCTSP_Circuit}
\end{figure}

\subsection{Problem statement}

In this paper, the QCTSP is assumed to be stored in three registers named the index register $\ket{0}_{index}$, the data register $\ket{0}_{data}$ and the time register $\ket{0}_{time}$ for the storage of the index sequence$\{k_j\}$, the path sequence $\{X_j\}$ and the holding time sequence $\{\tau_j\}$, respectively. 
Then the QCTSP state, defined as Definition \ref{def:1} and encoded in either Definition \ref{def:2} or Definition \ref{def:3}, lies in the Hilbert space $\mathcal{H}_{k,x}\bigotimes\mathcal{H}_t={(\mathbb{Z}_\mathcal{S}\mathbb{C}\mathbb{Z}_T)}^{\bigotimes n}$. 
More specifically, it can be represented as a superposition of possible paths
\begin{subequations}\label{eq:con1}
    \begin{equation}\label{eq:con1a}
        \ket{\psi}=\sum_{\bm{k,t} \in \mathbb{Z}_\mathcal{S}^n\otimes\mathbb{Z}_T^n}p(\bm{k,t})\ket{\bm{k,x(k),t}},
    \end{equation}
{where each path is stored in the state}
\begin{equation}\label{eq:con1b}
    {\ket{\bm{k,x(k),t}} = \ket{k_1}\ket{t_1}\ket{k_2}\ket{t_2}...\ket{k_n}\ket{t_n},}
\end{equation}
with possibility 
\begin{equation}\label{eq:con1c}
    \mathbb{P}(\bm{k,t}) = p(\bm{k,t})^2.
\end{equation}.
\end{subequations}
{Here the index vector $\bm{k} = (k_1, ..., k_n)^\mathsf{T} \in \mathbb{Z}_\mathcal{S}^n$ satisfies
$$x_{j, k_j} = x(k_j): 1\le j\le n, 1\le k_j \le \abs{\mathcal{S}},$$
and means that the $j^{th}$ piece takes the ${k_j}^{th}$ value in  the state space $\mathcal{S}$, 
and $\bm{t} = (t_1, ..., t_n)^\mathsf{T} \in \mathbb{Z}_T^n$
where $1\le t_j \le T$ is a realization of the sequence $(\tau_1, ..., \tau_n)$
and means that the $j^{th}$ piece stays for $t_j$ time slices bounded by $T$.
According to Definition \ref{def:1}, for each time piece, the current pair $(X_j, \tau_j)$ is derived from the past pieces $\{(X_{j'}, \tau_{j'}):1\le j'\le j-1\}$ with the possibility $\mathbb{P}(\bm{k,t})$  satisfying
\begin{subequations}
\begin{equation}
    \begin{split}
        \mathbb{P}(\bm{k,t})=&\mathbb{P}(X_1=x(k_1))\\
    \times&\mathbb{P}(\tau_1=t_1|X_1=x(k_1))\\
    \times&\prod_{j=2}^n \mathbb{P}(X_j=x(k_j)|X_{j'},\tau_{j'}=x(k_{j'}),t_{j'}:1\le j' < j)\\
    \times&\prod_{j=2}^n \mathbb{P}(\tau_j=t_j|X_{j'},\tau_{j'}=x(k_{j'}),t_{j'}:1\le j' < j\\
    &\text{ and } X_j=x(k_j)).
    \end{split}
\end{equation}
The corresponding possibility amplitude $p(\bm{k,t})$ satisfies
\begin{equation}\label{eq:con2}
    \begin{split}
        p(\bm{k,t})=&p(X_1=x(k_1))\\
    \times&p(\tau_1=t_1|X_1=x(k_1))\\
    \times&\prod_{j=2}^n p(X_j=x(k_j)|X_{j'},\tau_{j'}=x(k_{j'}),t_{j'}:1\le j' < j)\\
    \times&\prod_{j=2}^n p(\tau_j=t_j|X_{j'},\tau_{j'}=x(k_{j'}),t_{j'}:1\le j' < j\\
    &\text{ and } X_j=x(k_j)),
    \end{split}
\end{equation}
\end{subequations}
as a direct result.}
Besides, the possibility amplitudes also satisfy
\begin{equation}\label{eq:con3}
\sum_{\bm{k}, \bm{t}}p^2(\bm{k,t})=1.
\end{equation}

Given a CTSP as defined in Definition \ref{def:1}, our first goal is to derive the corresponding QCTSP Eq.~\eqref{eq:con1a} encoded by Definition \ref{def:2} and \ref{def:3}. More formally, the problem is defined as follows:

\begin{problem}
  \problemtitle{\textbf{Problem 1:} QCTSP Preparation \textit{(QCTSP-P)}}
  \probleminput{A CTSP defined in Definition \ref{def:1} and three registers $\ket{0}_{index}\ket{0}_{data}\ket{0}_{time}$}
  \problemquestion{Is there a procedure, that is, a unitary operator \emph{$\mathcal{U}$} acting on the registers such that with the notations described in Eq.~\eqref{eq:con1a} - Eq.~\eqref{eq:con1c} and summarized in Table \ref{tab:tab1}, 
$\mathcal{U} \ket{0}_{index}\ket{0}_{data}\ket{0}_{time} = \sum_{\bm{k,t} \in \mathbb{Z}_\mathcal{S}^n\otimes\mathbb{Z}_T^n}p(\bm{k,t})\ket{\bm{k,x(k),t}}$ and the possibility amplitude $p(\bm{k,t})$ satisfies Eq.~\eqref{eq:con2} and Eq.~\eqref{eq:con3}?}
\end{problem}

\noindent Supposed that a QCTSP has been prepared taking the form of Eq.~\eqref{eq:con1}, our second goal is to extract the information of interest, a quantity determined by the knowledge of the past path $\{X(t): T> t\ge0\}$. This problem can be stated formally as:

\begin{problem}
  \problemtitle{\textbf{Problem 2:} QCTSP Extraction  \textit{(QCTSP-E)}}
  \probleminput{A QCTSP defined as Eq.~\eqref{eq:con1a}-Eq.~\eqref{eq:con3} and a quantity $\mathcal{Q}$ determined by it.}
  \problemquestion{Is there a procedure, that is, a unitary operator \emph{$\mathcal{V}$} acting on the QCTSP such that
$\mathcal{V}\ket{\psi}$ is the desired quantity $\mathcal{Q}(\{X(t):T> t\ge0\})$?}
\end{problem}

{Measurements can be implemented at the end of QCTSP-E circuit to extract the quantum information into a classical register.
Despite this, it should be noticed that the information is stored in the quantum state before any measurements, 
and hence the QCTSP-E circuit can be utilized as a sub-module of another bigger circuit.}

As there are many classes of QCTSP to be prepared, together with many quantities of interest, the detailed quantum algorithms are expounded in section \ref{sec:prepare} and section \ref{sec:extract}, respectively. And a brief framework illustrating the basic idea is as follows.

\subsection{Framework}

The problem of modeling and analyzing the paths of QCTSP via a quantum processor can be divided into two main steps: Encode and prepare the QCTSP, and then decode and analyze it.

There are prominent challenges to overcome when we try to solve the \textit{QCTSP-P} problem. 
Intuitively, the complexity of the \textit{QCTSP-P} problem is determined by the number of pieces $n$, the average holding time $\tau_{avg}$, and the intrinsic evolution complexity.
The compressed encoding method appears to be natural and simple, and leads to reductions in the qubit number and the circuit depth (as illustrated in Figure  \ref{DSP_QCTSP_Circuit}). 
By introducing the indicator \textit{memory length} to depict the evolutionary complexity of CTSP, we first study the simplest case of memory-less process. 
An interesting observation of the memory-less process holding time is made, and this induces a constant circuit depth. 
As the memory length increases, the underlying QCTSP turns to be more complicated (shown in Figure  \ref{CTSP}) while the property of the holding time still works: 
Specific types of QCTSP can be prepared efficiently (see Figure  \ref{levy_picture}, Figure  \ref{markov_picture}, and Figure  \ref{cox_picture}), and the detailed construction and the corresponding proofs are given in section \ref{sec:prepare} and Appendix \ref{sec:7}, respectively. 
The qubit number is reduced from $O(T\ln{S})=O(n\tau_{avg}\ln{S})$ of the uniform sampling method to $O(n\ln{(\tau_{avg}S)})$ of the QCTSP method, making an exponential reduction on the parameter $\tau_{avg}$. 
And the circuit depth is optimized from $O(T\ln{nS})=O(n\tau_{avg}\ln{(nS\tau_{avg})})$ of the uniform sampling method to $O(n\ln{(nS)})$ of the QCTSP method, also making an exponential reduction of qubit number on the parameter $\tau_{avg}$. 
The results of state preparation are summarized in Table \ref{tab:tab2}.
\begin{table*}[h]
\centering
\caption[c]{{\textbf{Comparison of QCTSP and Uniform Sampling}
In this table, most often-used CTSP' s preparation procedure is summarized: the memory length, qubit number and circuit depth of each CTSP type, which have been given in section \ref{sec:prepare} and proved in Appendix \ref{sec:7}, are listed in the second, third, and fifth columns, respectively.}} \label{tab:tab2}
\centering
\resizebox{\textwidth}{!}{
\begin{tblr}{width=\linewidth,
            hlines=0.3pt, vlines=0.3pt,
            }
\SetCell[r=2]{c}\makecell[c]{QCTSP\\ Type}&\SetCell[r=2]{c}$\mathcal{ML}$&\SetCell[c=2]{c} Qubit Number&&\SetCell[c=2]{c}Circuit Depth&\\
&&\makecell[c]{QCTSP (our work)}&\makecell[c]{Uniform Sampling}&\makecell[c]{QCTSP (our work)}&\makecell[c]{Uniform Sampling}\\
\makecell{Poisson\\point\\process\\} & \makecell[c]{0} & $n\lceil\log{(\frac{-n\ln{\epsilon}}{\lambda})}\rceil$ & $\frac{-n\ln{\epsilon}}{\lambda}\lceil\log{(\frac{-n\ln{\epsilon}}{\lambda})}\rceil$ & $\lceil\log{n}\rceil$ & $\lceil\log{(\frac{-n\ln{\epsilon}}{\lambda})}\rceil$\\
\makecell[c]{compound\\Poisson\\process\\} & \makecell[c]{0} & $n\lceil\log{(\frac{-nS\ln{\epsilon}}{\lambda})}\rceil$ & $\frac{-n\ln{\epsilon}}{\lambda}\lceil\log{(\frac{-nS\ln{\epsilon}}{\lambda})}\rceil$ & $S+n\lceil\log{(nS)}\rceil$ & $S+\frac{-n\ln{\epsilon}}{\lambda}\lceil\log{(\frac{-nS\ln{\epsilon}}{\lambda})}\rceil$\\
\makecell[c]{L\'{e}vy\\process\\} & \makecell[c]{0} & $n\log{(nS')}$ & $T\log{S}$ & $S'+n\lceil\log{(nS')}\rceil$ & $S+T\log{(TS)}$\\ 
\makecell[c]{continuous\\Markov\\process\\} & \makecell[c]{1} & $n\lceil\log{(-\frac{S\ln{\epsilon}}{\lambda_{min}})}\rceil$ & $-\frac{n\ln{\epsilon}}{\lambda_{min}}\lceil\log{(-\frac{S\ln{\epsilon}}{\lambda_{min}})}\rceil$ & $nS^2$ & $-\frac{n\ln{\epsilon}}{\lambda_{min}}S^2$\\
\makecell[c]{\\$k$-order \\Markov\\process\\}  & \makecell[c]{$k$} & $n\lceil\log{(-\frac{S\ln{\epsilon}}{\lambda_{min}})}\rceil$ & $-\frac{n\ln{\epsilon}}{\lambda_{min}}\lceil\log{(-\frac{S\ln{\epsilon}}{\lambda_{min}})}\rceil$ & $nS^{k+1}$ & $-\frac{n\ln{\epsilon}}{\lambda_{min}}S^{k+1}$\\
\makecell[c]{Cox\\process\\} & \makecell[c]{0} & $n(q_\mathcal{F}+q_\mathcal{G}-\frac{\ln{\epsilon}}{\abs{\lambda}_{min}})$ & $-\frac{n\ln{\epsilon}}{\abs{\lambda}_{min}}(q_\mathcal{F}+q_\mathcal{G})$ & $\max{\{q_\mathcal{G}, q_\mathcal{F}+2^{d_\mathcal{F}}\}} + nq_\mathcal{G}$  & $\max{\{q_\mathcal{G}, q_\mathcal{F}+2^{d_\mathcal{F}}\}} - \frac{nq_\mathcal{G}\ln{\epsilon}}{\abs{\lambda}_{min}}$\\
\makecell[c]{general\\ CTSP} & \makecell[c]{$n$} & $n\log{\frac{ST}{n}}$ & $T\log{S}$ & $(\frac{ST}{n})^n$ & $S^T$\\
\end{tblr}}

\end{table*}
The even more challenging \textit{QCTSP-E} problem is to efficiently decode and extract information from QCTSP so that the quantum speed-up on preparation will not be diminished. 
By a controlled rotation gate-based Riemann summation of rectangles (as illustrated in Figure  \ref{extraction} (a) and Figure  \ref{extraction} (c)), the weightless summation of a discrete path can be extended to the time integral of a continuous path. 
In addition, a parallel coordinate transformation can be implemented to achieve a path-dependent weighted integral (as illustrated in Figure  \ref{extraction} (b), Figure  \ref{extraction} (d) and Figure  \ref{extraction} (e)). 
Furthermore, benefited from our encoding and extracting method, the discontinuous jumps and transitions can be detected and modeled more precisely. 
And history-sensitive information such as first hitting time can be extracted easily as a consequence. Moreover, by introducing the amplitude estimation algorithm, the information extraction QCTSP also admits a further quadratic speed-up. 
The detailed construction and the corresponding proofs are given in section \ref{sec:extract} and Appendix \ref{sec:8}.

\section{STATE PREPARATION}\label{sec:prepare}

\begin{figure}

\includegraphics[width=0.85\textwidth]{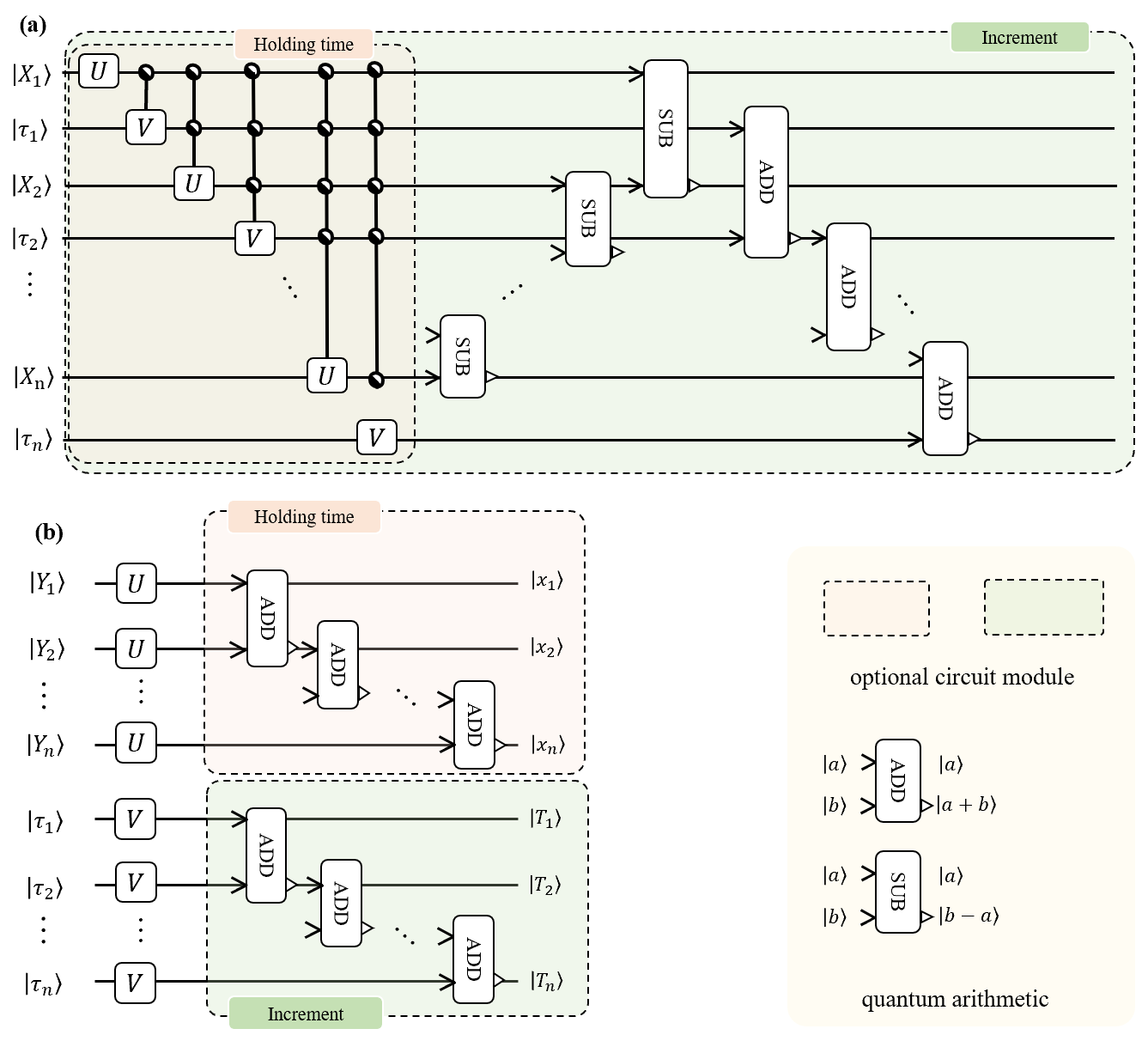}

\caption{\textbf{Quantum Circuits for General CTSP and compound Poisson Process} 
{In this figure, the circuits of QCTSPs with different memory lengths are given for comparison. For each QCTSP, the circuit can be divided into two parts: the state space part and the corresponding holding time part.
The state space part is recorded in a sequence of registers named $X_j$ or $Y_j$ for the $j^{th}$ piece,
wherein $Y_j = X_j - X_{j-1}$ denotes the increment for the increment representation of QCTSP. 
The corresponding holding time part is stored in the registers named $\tau_j$ for the $j^{th}$ holding time.
Intuitively, the structure and the gate complexity of a given QCTSP is closely linked to the memory length.
In the most complicated case of memory length $\mathcal{ML}(X)=n$, the current state is determined by all of the $n-1$ pieces. 
Hence, the $\ket{X_n}$ should be entangled with
all past registers via $2(n-1)$-controlled operators $U$.
Furthermore, the corresponding holding time $\tau_n$ is concerned with not only the past $n-1$ pieces but also the current state.
This is implemented through $2n-1$-controlled operators $V$.
By contrast, in the most simple case of memory length $\mathcal{ML}(x)=0$ (for example, the well-known compound Poisson process), the state spaces and the corresponding holding times are supposed to be independent.
As a result, each register can be prepared via independent implementations of operators $U$ and $V$.
From a higher perspective, the structure of the quantum circuit is intimately linked to the evolutionary dynamics of the underlying CTSP. As the characterization $\mathcal{ML}(X)$ increases, the long-range entanglement also grows. Consequently, this expansion necessitates the introduction of more multi-control unitary operators and SWAPs (taking into account the possible topology limitations of the quantum processor). As a result, the gate complexity of the quantum circuit experiences growth.}
More specifically, one have:
\textbf{(a)} The preparation circuit of the general CTSP is shown and the qubits storing $X_j$ and $\tau_j$ are highly entangled. 
The circuit is too complicated to simulate even for the most simple case where $\mathcal{S}=\{0,1\}$ and $T=2$. For the $j^{th}$ step, there are $2^{2j-2}$ $U_j$ operators and $2^{2j-1}$ $V_j$ operators. 
Hence the total number of unitary operators turns to be $2^{2n}-1$ and increases exponentially. 
Besides, the growth of the operator size leads to enormous challenges of gate decomposition and long-term entanglement. 
\textbf{(b)} The preparation of the Poisson process and compound Poisson process is shown. 
The increments $Y_j$ and holding time $\tau_j$ are I.I.D and can be prepared parallel through $U$ and $V$ operators. 
The alternative encoding methods can be evaluated through a sequence of add operators on $Y_j$ and $\tau_j$ registers in the two boxes, with respect to holding time representation and increment representation.}
\label{CTSP}
\end{figure}

To understand why the state preparation procedure of QCTSP is non-trivial, one should notice that the sequence of pairs $(X_j,\tau_j)$ can be entangled with each other for $1\le j\le n$. 
Intuitively, the more the current state is influenced by past information, the more complicated gates and deeper circuits are needed.  
More precisely, the memory length of a given CTSP is defined as 
\begin{equation}
\begin{split}
\mathcal{ML}(X(t))=\min{}\{n|&\mathbb{P}[X(t_0)|X(t):0 \le t < t_0]=\\
&\mathbb{P}[X(t_0)|X(t):t_0-n\Delta t\le t < t_0]\\
&\text{for all }t_0\ge n\Delta t\},
\end{split}
\end{equation}
which means that the conditional probability distribution of the current time $t_0$ is totally determined by the information in the latest $\mathcal{ML}(X)$ time slices. 
Two extreme situations are compared in Figure  \ref{CTSP}: 
One is the worst case where the memory length $\mathcal{ML}(X)=n$ is as long as the number of pieces. The other is the opposite case where the memory length $\mathcal{ML}(X)=0$, and $X(t)$ is said to be memory-less in this case.
For the first case of the longest memory length,  the present pair registers $(X_j, \tau_j)$ should be entangled with all past time pairs of variables $\{(X_{j'}, \tau_{j'}):1\le j'\le j-1\}$ together (as illustrated in Figure  \ref{CTSP} \textbf{(a)}), and this surely leads to a disaster of circuit depth and gate complexity.
The second case, on the contrary, turns out to be quite simple, and we have the following result:
\begin{thm} \label{thm:1}
\textbf{(Memory-less Process' Holding Time)} Supposing that $\{X(t):0\le t\le T\}$ is memory-less, then the holding time $\tau_j$ for any state $X_j$ follows an exponential distribution $\mathbb{P}[\tau_j \ge t]=e^{-\lambda_j t}$ with $\lambda_j$ determined by the current sate $X_j$. Moreover, given an $\epsilon$ such that the cumulative density function(c.d.f) $f$ satisfies $\mathbb{P}[\tau_j>T]=1-f(T)=e^{-\lambda T}<\epsilon$, $\tau_j$ can be prepared via a circuit consisting of $\lceil\log{(-\frac{1}{\lambda_j}\ln{\epsilon})}\rceil$ qubits with constant circuit depth $1$, and the gate complexity is $\lceil\log{(-\frac{1}{\lambda_j}\ln{\epsilon})}\rceil$.
(see proof in Appendix \ref{pf:1})
\end{thm}

This simple result should not be surprising since, as discussed above, the memory-less condition $\mathcal{ML}(\tau_j)=0$ leads to no-entanglement states and easier preparation.
Following this thought, most of the often-used CTSPs have been systematically studied and efficiently prepared. 
The conclusion that the complexity of the state preparation of QCTSP tends to be higher as the memory length grows is further evidenced by those results summarized in table \ref{tab:tab2}, leaving the detailed explanation and relevant applications given as follows.

\subsection{L\'{e}vy process} \label{sec:levy}
The first and essential family of CTSPs is the L\'{e}vy process, which is one of the most well-known families of continuous-time stochastic processes, including the Poisson process and the Wiener process (also known as Brownian motion), with applications in various fields {(For more details, refer to \cite{applebaum2004levy} and the references therein).} 
Formally, $\{X(t):t\ge0\}$ is said to be a L\'{e}vy process if 1) the increments $\{X(T_{j+1})-X(T_j):1\le j \le n\}$ are independent for any $0< T_1\le T_2 \le ... \le T_n$, and 2) stationary which means that $X(t) - X(s)$ depends only on $t-s$ and hence is equal in distribution to $X(t-s)$(as illustrated in Figure \ref{levy_picture} \textbf{(a)} and \textbf{(b)}). 
\begin{figure}[ht]
\includegraphics[width=0.95\textwidth]{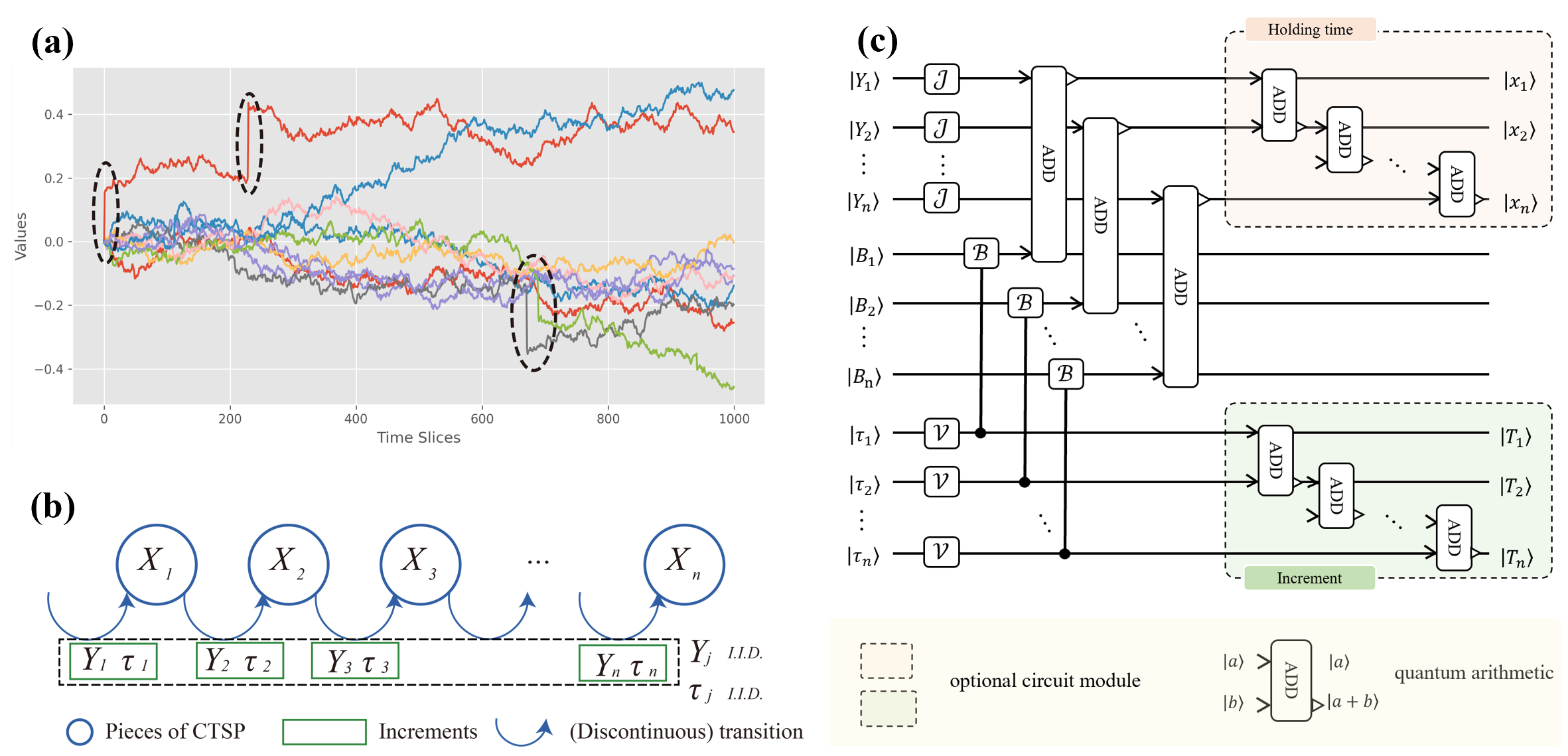}
\caption{\textbf{Prepare a L\'{e}vy Process.}  
In this figure, the preparation of a general L\'{e}vy Process is given: 
\textbf{(a)} $10$ simulated L\'{e}vy process paths of $1000$ steps: 
The intensity of discontinuous jump is $\lambda=1$, and the standard deviation of jump size is $v=0.3$. 
The standard deviation and the drift rate of the Brownian part are $\sigma=0.2$ and $a=0.02$, respectively. 
As circled in the picture, discontinuous jumps happen at random timestamps that are hard for a uniform sampling method to capture.
\textbf{(b)} A L\'{e}vy Process is embedded into discrete stochastic processes via holding time representation and increment representation encoding methods. 
As $\mathcal{ML}(X)=0$, the corresponding stochastic pairs $(Y_j, \tau_j)$ are I.I.D. variables. 
\textbf{(c)} The preparation of the general L\'{e}vy process is shown. 
{Intuitively, the dynamic of L\'{e}vy process is simple as the memory length is zero and hence each piece evolves independently.
Furthermore, the evolution can be decomposed into two distinct parts: a continuous Brownian motion and a discontinuous jump.}
The discontinuous jumps $J_j$ and holding times $\tau_j$ are I.I.D, and can be prepared parallel through $J$ and $V$ operators. 
The Brownian motion components $B_j$ are determined by the holding time $\tau_j$, and can be derived by parallel controlled-$B$ operators. 
The increments $Y_j = J_j + B_j$ are then computed by parallel adder operators. 
It should be mentioned that operators in the grey box are implemented parallel within constant circuit depth. The alternative encoding methods can be evaluated through a sequence of adder operators on $Y_j$ and $\tau_j$ registers in the two boxes, with respect to holding time representation and increment representation.}
\label{levy_picture}
\end{figure}
Under these assumptions, a L\'{e}vy process for the most general case can be prepared efficiently, and we have the following result:

\begin{thm}\label{thm:2}
\textbf{(Upper bound of general L\'{e}vy Process' preparation)}  Suppose that $\{X(t):0<t<T\}$ is a L\'{e}vy process, 
$\Delta t$ is the length of step, 
$n=\frac{T}{\Delta t}$ is the length of sample path, 
and $S=\abs{\mathcal{F}}$ is the number of discretization intervals of the underlying distribution $\mathcal{F}$. 
Then the L\'{e}vy process can be prepared on $n\log{(nS)}$ qubits within $O(S+n\log{(nS)})$ circuit depth, and the gate complexity is thus $O(n(S+\log{(nS)}))$. 
(see proof in Appendix \ref{pf:2})
\end{thm}

It should be mentioned that the preparation procedure's complexity is mainly determined by two  factors: the sampling times $n=\frac{T}{\Delta t}$ and the preparation complexity of the underlying distribution $\mathcal{F}$. 
In the most general case where compensated generalized Poisson process with countably many small jump discontinuities and Wiener process are taken into consideration, the sample space size $S=\abs{\mathcal{F}}$ is large and the preparation procedure of the underlying distribution $\mathcal{F}$ may consume an enormous amount of classical computation resource. Hence it is requisite for one to study the more specific situations besides the upper bound given in Theorem \ref{thm:2}, and some optimized preparation results are given below. 
The first example is the \textit{Poisson point process} that plays an essential role in modeling the arrival of independent random events. 
More precisely, despite several equivalent definitions for different domains and applications, a CTSP is a \textit{Poisson point process} on the positive half-line if the increment $Y_j$ is a constant $1$ and the underlying distribution of each holding time $\tau_j$ between the events is an exponential distribution that can be efficiently prepared as shown in Theorem \ref{thm:1}. 
As a direct consequence, one has:
\begin{coro}\label{thm:3}
\textbf{(Poisson Point Process' preparation)} Supposing that $\{X(t):0<t<T\}$ is a Poisson process, $\Delta t$ is the length of step, $n=\frac{T}{\Delta t}$ is the length of sampling times, and $\epsilon$ is the truncation of the exponential distribution $\mathcal{F}(\lambda)$. Then the Poisson process can be prepared on $n\lceil\log{(\frac{-n\ln{\epsilon}}{\lambda})}\rceil$ qubits with $O(n\lceil\log{(\frac{-n\ln{\epsilon}}{\lambda})}\rceil)$ control-rotation gates, and the circuit depth is $\lceil\log{n}\rceil$. 
(see proof in Appendix \ref{pf:3})
\end{coro}

\noindent Another more flexible example is the \textit{compound Poisson process} as a generalization, i.e., the jumps' random arrival time follows a Poisson process and the size of the jumps is also random, following an underlying distribution $\mathcal{G}$. 
It should be noticed that each jump's sampling space size $S$ of \textit{compound Poisson process} is usually fixed in practical, apparently different from the general L\'{e}vy process discussed above. 
We then get the following result with some modification on the circuit(as illustrated in Figure  \ref{CTSP} \textbf{(b)}):
\begin{coro}\label{thm:4}
\textbf{(Compound Poisson Process' preparation)} Suppose that $\{X(t):0<t<T\}$ is a compound Poisson process with $n=\frac{T}{\Delta t}$ being the length of sample path, $S$ denoting the sampling space size of each point, and $\epsilon$ being the truncation of the exponential distribution $\mathcal{F}(\lambda)$, then the compound Poisson process can be prepared on $n\lceil\log{(\frac{-nS\ln{\epsilon}}{\lambda})}\rceil$ qubits. The circuit depth is $S+n\lceil\log{(nS)}\rceil$, and the gate complexity is $O(n\lceil\log{(-\frac{nS\ln{\epsilon}}{\lambda})}\rceil)$.
(see proof in Appendix \ref{pf:3})
\end{coro}

According to the L\'{e}vy-It\^{o} decomposition, a general L\'{e}vy process can be decomposed into three components: 
a Brownian motion with drift $\sigma B_t + at$, 
a compound Poisson process $Y_t$, 
and a compensated generalized Poisson process $Z_t$ as follows:
\begin{equation}
X_t = \sigma B_t + at + Y_t + Z_t,
\end{equation}
where $\sigma$ is the volatility of the Brownian motion, $a$ is the drift rate, $Y_t$ is the compound Poisson process of jumps larger than 1 in absolute value, and $Z_t$ is a pure jump process with countably (infinite) many small jump discontinuities. 
As $Z_t$ contains infinitely many jumps as a zero-mean martingale, it is hard and unnecessary for preparation. 
One can prepare a general L\'{e}vy process without compensated parts as follows: 
Firstly, the sequence of exponential jumps $J_j$ and holding times $\tau_j$ are prepared via $\mathcal{J}$ and $\mathcal{V}$ operators, respectively. 
Secondly, the Brownian motions $B_j$ are derived by Grover's state preparation method through $\mathcal{B}$ operators controlled by holding time $\tau_j$. 
Thirdly, an adder operator is introduced for the sum of pure jump and Brownian parts. 
It should be mentioned that the drift term $at$ has been omitted since any function on this linear term can be easily translated into a function and hence well-evaluated by Theorem \ref{thm:8} and Theorem \ref{thm:9}. The circuit is shown in Figure \ref{levy_picture} \textbf{(c)}.

\subsection{Continuous Markov process} 
The second and more complicated family of CTSPs is the \textit{continuous Markov process}. 
In brief, a stochastic process is called a \textit{continuous Markov process} if it satisfies the Markov condition: 
For all $t\ge0$, $s\ge0$, $\mathbb{P}[X(s+t)=j|X(s)=k \land X(u):0\le u <s]=\mathbb{P}[X(s+t)=j|X(s)=k]$, where $j,k \in\mathcal{S}$. 
\begin{figure}
\includegraphics[width=0.95\textwidth]{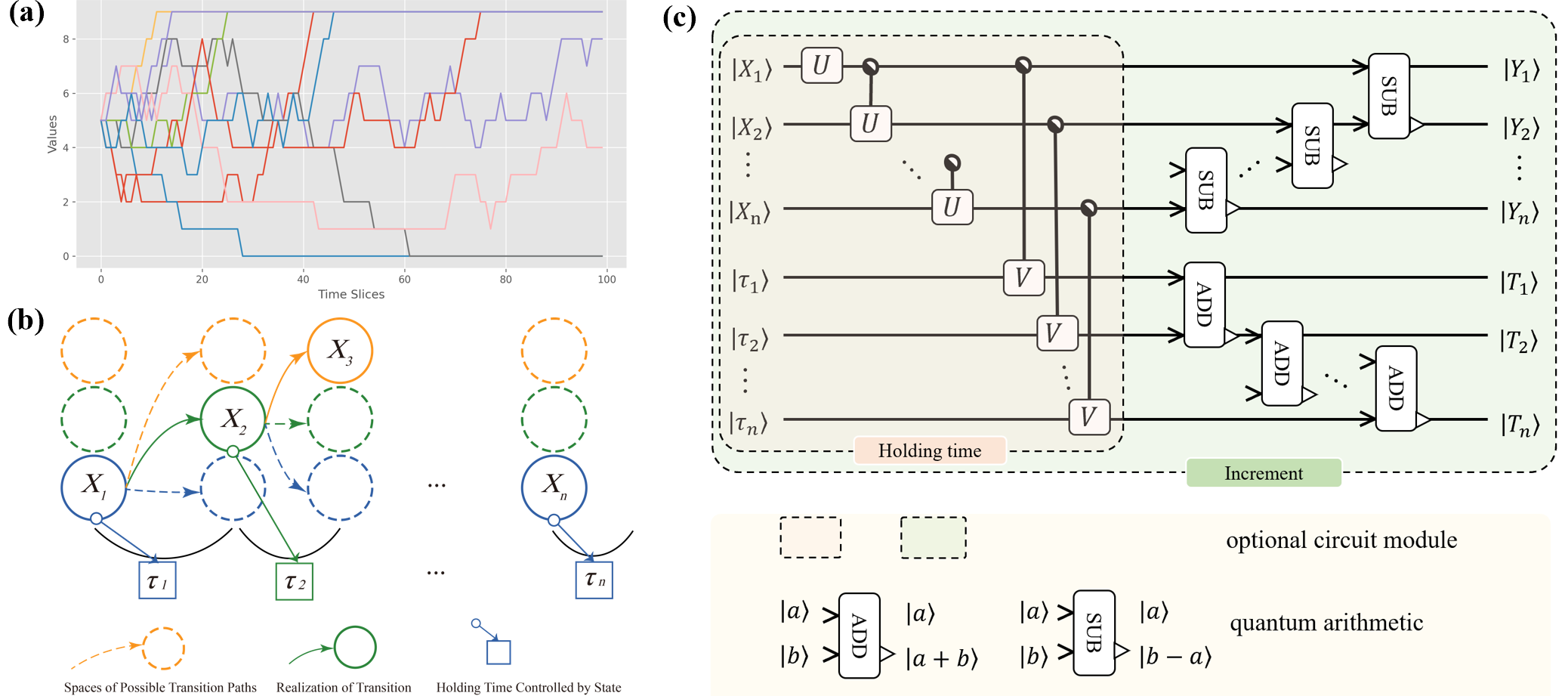}
\caption{\textbf{Prepare a Continuous Markov Process.}  In this figure, the preparation of a continuous Markov Process is given: 
\textbf{(a)} $10$ simulated \textit{continuous Markov process} paths of $100$ steps: The  transition rates from $\ket{n}$ to $\ket{n-1}$ and $\ket{n+1}$ are both $0.05$,  and the state $\ket{0}$ and $\ket{10}$ are absorbing states. Since the possibility of staying in the current state is $p=0.9$, the continuous Markov process tends to be in-variant for a long holding time, and thus it can be compressed a lot via our encoding method.
\textbf{(b)} The \textit{continuous Markov process} is embedded via both holding time representation and increment representation. Due to $\mathcal{ML}(X)=1$, each piece's state space $X_j$ is controlled by the previous one $X_{j-1}$ while the holding time $\tau_j$ is determined by the current state $X_j$ through an exponential distribution.
\textbf{(c)} The preparation procedure of a \textit{continuous Markov process} is shown. The embedded discrete Markov process $X_j$ is prepared by a sequence of $U$ and controlled-$U$ operators, while the holding time $\tau_j$ is then determined by $X_j$ and prepared by controlled-$V$ operators. The increment representation is then prepared by an alternative sequence of adder and subtractor subcircuits in the box.}
\label{markov_picture}
\end{figure}
Intuitively speaking, given the present state $X(s)$, the future $\{X(s+t):t\ge0\}$ is independent of the past history $\{X(u):0\le u < s\}$, and the memory length of a \textit{continuous Markov process} is $\mathcal{ML}(X_t)=1$ (shown in Figure \ref{markov_picture} \textbf{(a)} and \textbf{(b)}). 
A \textit{continuous Markov process} can be efficiently prepared mainly for two reasons: the holding time $\tau_j$ only depends on the current state $X_j$, while the state sequence $\{X_j\}$ is independent from $\{X_{j'}:1\le j'< j-1\}$. 
As a result, the \textit{continuous Markov process} can be embedded into an independent discrete Markov chain depicting the transition probability, together with a sequence of controlled exponential distributions carrying the holding time information (as illustrated in Figure  \ref{markov_picture} \textbf{(c)}).
More precisely, we have the following result:
\begin{thm}\label{thm:6}
\textbf{(Continuous Markov Process' Preparation)}Supposing a stochastic process $\{X(t):t\ge0\}$ with discrete state space $\mathcal{S}\subseteq\mathbb{Z}$ to be a continuous-time Markov chain. 
Then it can be prepared on $n\lceil\log{(-\frac{S}{\lambda_{min}}\ln{\epsilon})}\rceil$ qubits via a quantum circuit consisting of $O(nS\lceil S+\log{(-\frac{\ln{\epsilon}}{\lambda_{min}})}\rceil)$ gates and $nS^2$ circuit depth, where $S=\abs{\mathcal{S}}$, $\lambda_{min}$ is the minimum of $\lambda_j$, and $\epsilon$ is the truncation error bound for time.
(see proof in Appendix \ref{pf:5})
\end{thm}

\subsection{Cox process}
\begin{figure}[ht!]
\includegraphics[width=0.95\textwidth]{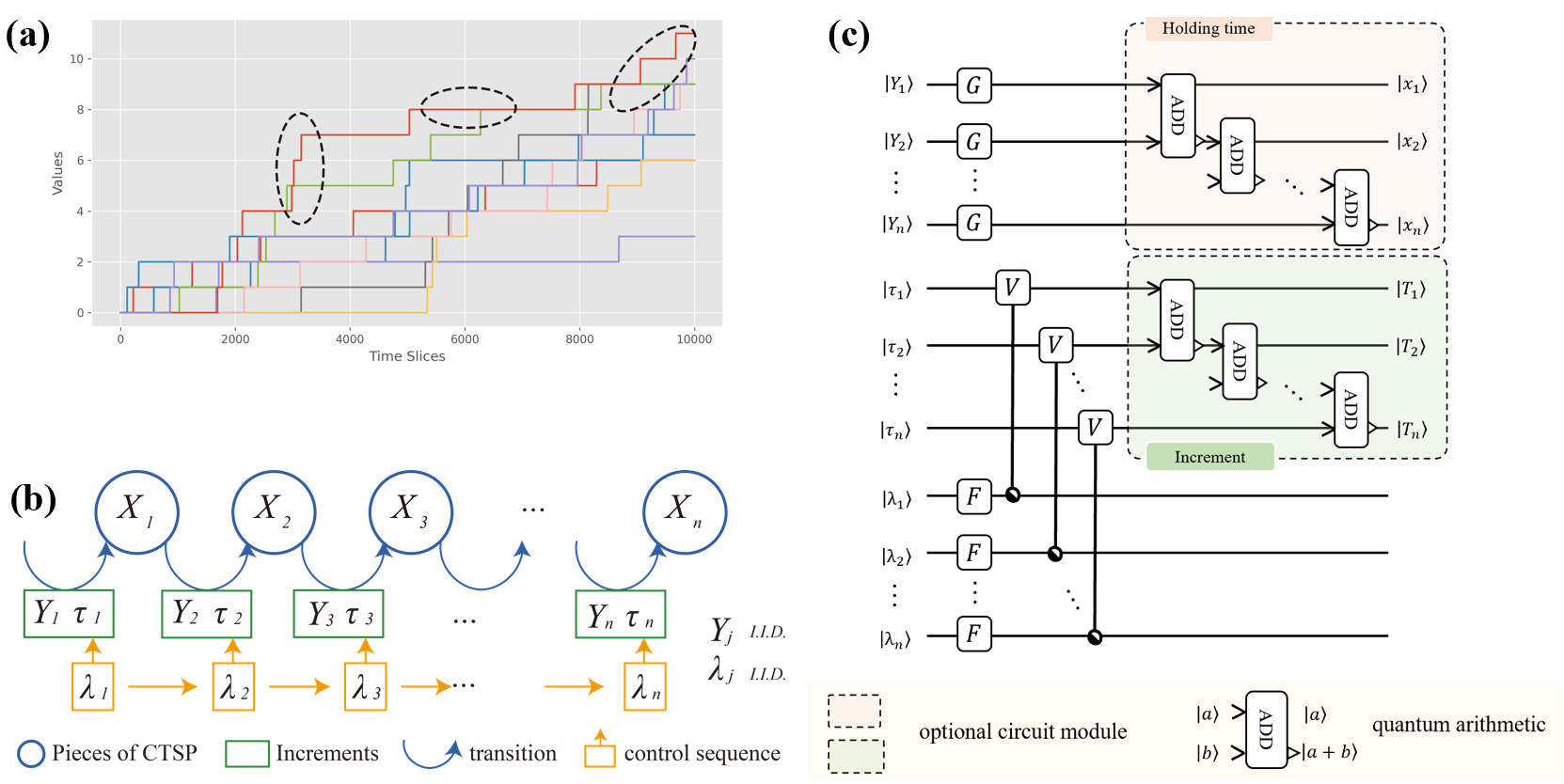}
\caption{\textbf{Prepare a Cox Process.}  In this figure, the preparation of a continuous Markov Process is given: 
\textbf{(a)} $10$ simulated Cox process paths of length $10$ and $10000$ steps: The varying intensity follows an exponential distribution with $\lambda=0.1$, and the size of the discontinuous jump is assumed to be a constant $1$. As circled in the picture, the intensity changes with time and the jump rate varies from high to low randomly.
\textbf{(b)} A Cox Process is again embedded via holding time representation and increment representation encoding methods. The process is more complicated as the holding time $Y_j$ of interest is controlled by another stochastic variable $\lambda_j$ and the increment $Y_j$ and the control variable $\lambda_j$ are independent identical distribution variables, respectively.
\textbf{(c)} The preparation of a Cox process is shown. Firstly, the stochastic control sequence $\lambda_j$ is prepared by $F$ operators, and then the holding time sequence $\tau_j$ is prepared by parallel controlled-$V$ operators. The increments $Y_j$ are I.I.D. and can be prepared through $G$ operators. Alternative subcircuits of holding time and increment representations follow as shown in the two boxes.}
\label{cox_picture}
\end{figure}
As further evidence illustrating the flexibility and generalizability of our encoding method, the Cox process, a useful framework for modeling prices of financial instruments in financial mathematics \cite{lando1998cox}, can be efficiently prepared. 
One of the most significant distinctions between the Cox process and those processes discussed above is that both the states and the evolutionary law vary randomly over time (as illustrated in Figure  \ref{cox_picture} \textbf{(a)} and \textbf{(b)}). 
As a direct consequence, the preparation procedure and the corresponding circuit should be dynamic and flexible to characterize this doubly stochastic property (shown in Figure  \ref{cox_picture} \textbf{(c)}). 
Precisely speaking, a Cox process is a generalization of the Poisson process where the intensity $\lambda(t)$ itself is also a stochastic variable of some distribution $\mathcal{F}$. 
The preparation procedure can be divided into two steps: 
First, the stochastic control sequence of $\lambda_j$ is prepared by some $F$ operator for the underlying distribution $\mathcal{W}$. Then the holding time sequence $\tau_j$ is prepared by parallel controlled-$V$ operators for the underlying exponential distribution. 
The operators of different exponential distributions can be implemented by a sequence of controlled rotation gates of different rotation angles. 
The increments $Y_j$ are assumed to be I.I.D. and can be prepared through some $G$ operators (common choice of distributions can be found in the appendix, see Figure  \ref{exponential_distribution_circuit} and Figure  \ref{eldc} for reference). 
In summary, we have the following result:
\begin{thm}\label{thm:7}
\textbf{(Cox Process' Preparation)}Suppose that $\{X(t):t\ge0\}$ is a Cox process with a varying intensity variable $\lambda(t)$, and $\lambda(t)$ follows an underlying statistical distribution $\mathcal{F}$ that can be prepared on  $q_\mathcal{F}$ qubits with circuit depth $d_\mathcal{F}$. 
The increments $Y_j$ follows an I.I.D. that can be prepared on $q_\mathcal{G}$ qubits with circuit depth $d_\mathcal{G}$. 
Also it is supposed that the minimum value of $\lambda(t)$ is $\lambda_{min}$ and the error bound is $\epsilon$, then it can be efficiently prepared on $O(n(q_\mathcal{F}+q_\mathcal{G}-\frac{\ln{\epsilon}}{\abs{\lambda_{min}}}))$ qubits with circuit depth $O(\max{\{q_\mathcal{G}, q_\mathcal{F}+2^{d_\mathcal{F}}\}} + nq_\mathcal{G})$ (holding time representation ) or $O(\max{\{q_\mathcal{G}, q_\mathcal{F}+2^{d_\mathcal{F}}\}} - \frac{n\ln{\epsilon}}{\lambda_{min}})$ (increment representation ).
(see proof in Appendix \ref{pf:6})
\end{thm}
 
 As the details of the QCTSP preparation procedure have been given above, the results are summarized in table \ref{tab:tab2}. From a high-level view, the computation and storage resource of the preparation procedure is totally determined by the time length $T$ and the memory length $\mathcal{ML}(X)$ of the underlying QCTSP $X(t)$. By our compressed encoding method and the corresponding preparation procedure, the number of copies of the state space is reduced from $O(T)=O(n\tau_{avg})$ of the uniform sampling method to $O(n\ln{\tau_{avg}})$ of the QCTSP method, leading an exponential reduction of qubit number on the parameter $\tau_{avg}$. By the encoding method and the observation on the holding time of the memory-less QCTSP, the circuit depths can be optimized from $O(T\ln{nS})=O(n\tau_{avg}\ln{(nS\tau_{avg})})$ of the uniform sampling method to $O(n\ln{(nS)})$ of the QCTSP method, also making an exponential reduction of qubit number on the parameter $\tau_{avg}$.

\section{INFORMATION EXTRACTION}\label{sec:extract}
Besides the state preparation problem, another challenge of QCTSP is to decode the compressed process and rebuild the  information  of interest. 
To extract desired information from QCTSP, one needs to evaluate its continuous counterpart, i.e., the integral of random variable $X$ on time $t$, instead of the discrete summation $\mathbb{E}[f(S_n)]=\sum_{\mathbf{j}\in K^n}f(\text{sum}\{\mathbf{x(j)}\})\mathbb{P}[\mathbf{X}=\mathbf{x(j)}]$ as studied in \cite{blank2021quantum}. 

\subparagraph{Weightless Integral.} 
\begin{figure}[hp]
\includegraphics[width=0.95\textwidth]{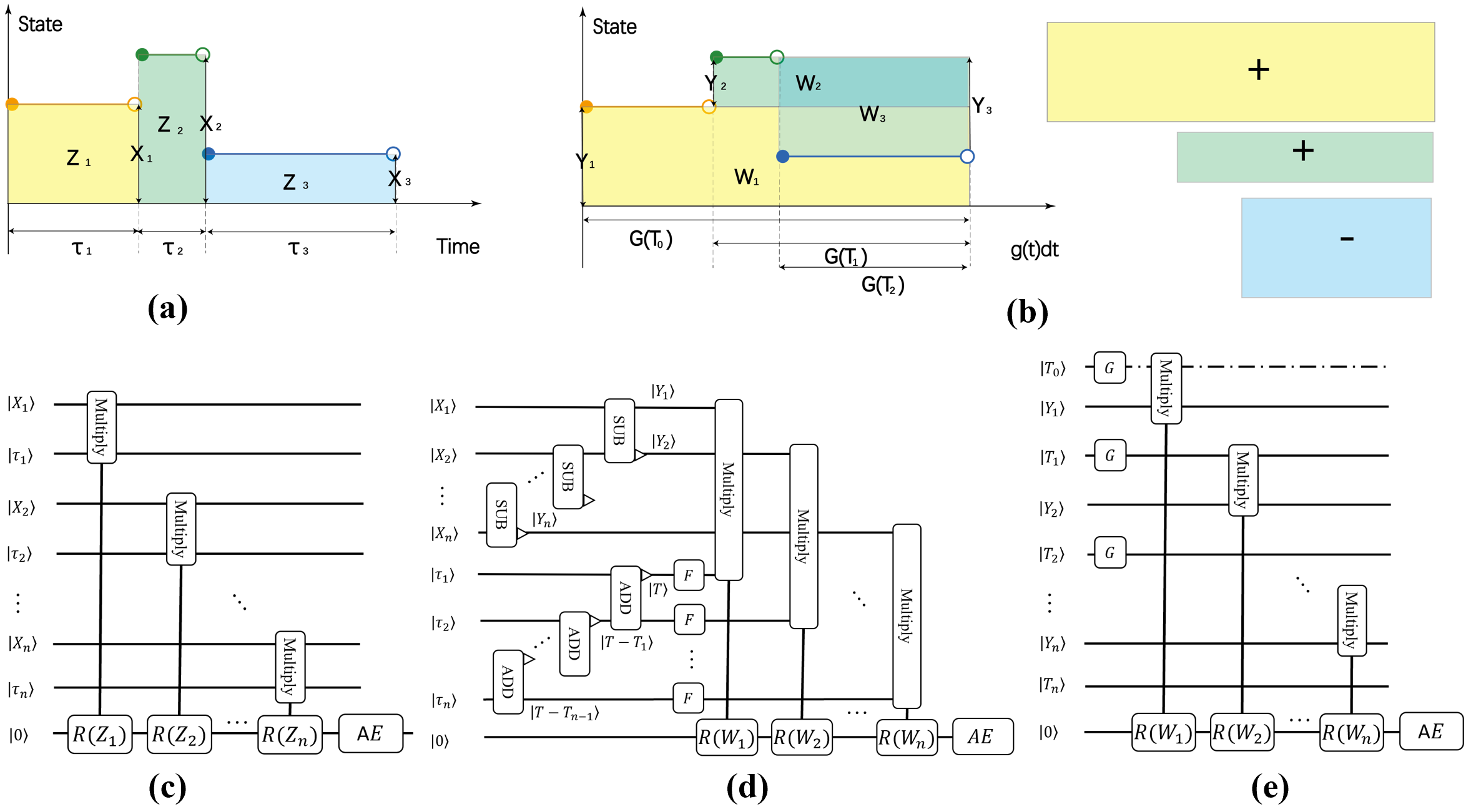}
\caption{\textbf{Information Extraction for Quantum Continuous Time Stochastic Process.} \textbf{(a)} The integral $I(X,T)=\int_{t=0}^TX(t)\,dt$ is represented as a summation of the area of rectangles derived from vertical division $I=\sum Z_j$. 
\textbf{(b)} By a coordinate transformation $t\rightarrow g(t)$, the integral  $J(X,T)=\int_{t=0}^Tg(t)X(t)\,dt$  can be represented as a summation of the directed area of rectangles derived from horizontal division $J=\sum W_j$: the yellow and green boxes both correspond to positive increments $Y_{1,2}$ and area $W_{1,2}$, while the blue box corresponds to a negative increment $Y_3$ and $W_3$.
\textbf{(c)} The circuit for evaluating $I(t)$ is shown. 
Supposing that $X_j$ and $\tau_j$ have been prepared as the input state, $n$ quantum multipliers are introduced to evaluate $Z_j=X_j\tau_j$: Rotation operators controlled by $Z_j$ are implemented on the target qubit, followed by a standard amplitude estimation subcircuit in the red box.
\textbf{(d)} The circuit for evaluating $J(t)$ of holding time representation is shown. Supposing that $X_j$ and $\tau_j$ have been prepared as the input state, then a sequence of subtractor operators is introduced to derive the increments $Y_j$, and another sequence of adder operators is introduced to evaluate the cumulative time from ending time $T-T_{j-1} = \sum_{j'=j}^n \tau_j'$). 
Hence $\int^T_{T_{j-1}}g(t)\,dt=F(T-T_{j-1})$ can be evaluated through the integral operators $F$. 
Following that are the desired variables $W_j=Y_jF(T-T_{j-1})$ by quantum multiplier operators: Rotation operators controlled by $W_j$ are implemented on the target qubit, followed by a standard amplitude estimation subcircuit in the red box.
\textbf{(e)} The circuit for evaluating $J(t)$ of increment representation is shown. Supposing that $Y_j$ and $T_j$ have been prepared as the input state, one more qubit is introduced to denote the beginning time $T_0$, followed by unitary operators $G$ act on $T_j$ to derive $G(T_{j-1})$. 
Then a sequence of multiplier operators are implemented to compute the desired product $W_j=Y_j G(T_{j-1})$: 
Rotation operators controlled by $W_j$ are implemented on the target qubit, followed by a standard amplitude estimation subcircuit in the red box. The circuit's depth is  less than in subfigure \textbf{(d)} as a consequence of increment representation's advantage on integral calculation.}
\label{extraction}
\end{figure}
More specifically, the expectation value of an integrable function $f:\mathbb{R}\rightarrow\mathbb{R}$ of the random variable
\begin{equation}
I(X,T)=\int_{t=0}^TX(t)\,dt
\end{equation}
can be efficiently evaluated via a framework developed as follows: First of all, the integral can be divided into vertical boxes (as shown in Figure \ref{extraction} \textbf{(a)}) as a Riemann summation 
$$I(X,T)=\int_{t=0}^TX(t)\,dt=\sum_{j=0}^nX_j\tau_j.$$
Secondly, the area $Z_j$ of each box is translated into multiplication of the holding time representation encoding variables $X_j$ and $\tau_j$. Benefiting from the encoding method, this expression is quite simple (as illustrated in Figure \ref{extraction} \textbf{(c)}).
Thirdly, the expected value of the integrable function $f(\sum_{j=0}^nX_j\tau_j)$ can be evaluated through a Fourier approximation $f_{P,L}(x)=\sum_{l=-L}^L c_l e^{i\frac{2\pi l}{p}x}$ as a $P-periodic$ function of order $L$:
\begin{equation*}
\sum_{l=-L}^L c_l (\mathbb{E}[\cos{(\frac{2\pi l}{p}(\sum_{j=1}^nZ_j))}]+i\mathbb{E}[\sin{(\frac{2\pi l}{p}(\sum_{j=0}^nZ_j))}]).
\end{equation*}
Each term $\mathbb{E}[\cos{(\frac{2\pi l}{p}(\sum_{j=0}^nZ_j))}]$ and $\mathbb{E}[\sin{(\frac{2\pi l}{p}(\sum_{j=0}^nZ_j))}]$ can be evaluated via rotation gates on the target qubit controlled by $Z_j$ followed by a standard amplitude estimation algorithm. Leaving the detailed proof in the appendix, one has:
\begin{thm}\label{thm:8}
\textbf{(Evaluating I(t))}Suppose that a QCTSP $\{X(t):t\ge0\}$ is given via holding time representation with $n$ steps, $l_x, l_\tau$ qubits for each pair $X_j, \tau_j$, and $\mathcal{P}$ circuit depth. Also, one has that $f_{P,L}(x)$ is the $P-periodic$ $L-order$ Fourier approximation of the desired integrable function of the random variable $f:\mathbb{R}\rightarrow\mathbb{R}$. 
Then given the error $\epsilon$, the expectation value of $f(I(X,T))=f(\int_{t=0}^TX(t)\,dt)$ can be efficiently evaluated within $O(\mathcal{P}+nl_xl_\tau)$ circuit depth and $O(\frac{LP}{\epsilon}(\mathcal{P}+nl_xl_\tau))$ time complexity.
(see proof in Appendix \ref{pf:7})
\end{thm}
\subparagraph{Weighted Integral.} Moreover, this framework applies to the more complicated and generalized situations where the time structure is considered and thus has no discrete correspondence. By introducing a time-dependent function $g(t)$, we can evaluate the expectation value of an integrable function $f:\mathbb{R}\rightarrow\mathbb{R}$ of the random variable
\begin{equation}
J(X,T)=\int_{t=0}^Tg(t)X(t)\,dt.
\end{equation}
Since the time structure $g(t)$ is considered, the integral is divided into horizontal directed boxes(as shown in Figure \ref{extraction} \textbf{(b)}):
$$J(X,T)=\sum_{j'=1}^nY_{j'}(\sum_{j=j'}^nG_j)=\sum_{j=1}^nW_j.$$
Due to a trade-off on circuit depth and qubits number(see Figure \ref{extraction} \textbf{(d)} and Figure \ref{extraction} \textbf{(e)} for reference), this summation can be evaluated via either holding time representation or increment representation. A similar Fourier approximation and amplitude estimation are employed to compute
$$\sum_{l=-L}^L c_l (\mathbb{E}[\cos{(\frac{2\pi l}{p}(\sum_{j=1}^nW_j))}]+i\mathbb{E}[\sin{(\frac{2\pi l}{p}(\sum_{j=0}^nW_j))}]),$$
and hence the following result is given:
\begin{thm}\label{thm:9}
\textbf{(Evaluating J(t))} 

\noindent\textbf{i)} Suppose a QCTSP $\{X(t):t\ge0\}$ of holding time representation given with $n$ steps, $l_x, l_\tau$  qubits for each pair  $X_j, \tau_j$, and $\mathcal{P}$ circuit depth. $g(t)$ is a function whose integral can be efficiently prepared with $\mathcal{F}$ circuit depth. Also, one has that $f_{P,L}(x)$ is the $P-periodic$ $L-order$ Fourier approximation of the desired integrable function of the random variable $f:\mathbb{R}\rightarrow\mathbb{R}$. Then given the error $\epsilon$, the expectation value of $f(J(X,T))=f(\int_{t=0}^Tg(t)X(t)\,dt)$ can be efficiently evaluated within $\mathcal{P}+\mathcal{F}+nl_xl_\tau+n\max{(l_x,l_\tau)}$ circuit depth and $O(\frac{LP}{\epsilon}(\mathcal{P}+\mathcal{F}+nl_xl_\tau+n\max{(l_x,l_\tau)}))$ time complexity.

\noindent\textbf{ii)} Suppose a QCTSP $\{X(t):t\ge0\}$ of increment representation given with $n$ steps, $l_Y, l_T$  qubits for each pair  $Y_j, T_j$, and $\mathcal{P}$ circuit depth. $g(t)$ is a function whose integral can be efficiently prepared with $\mathcal{G}$ circuit depth. Also, one has that $f_{P,L}(x)$ is the $P-periodic$ $L-order$ Fourier approximation of the desired integrable function of the random variable $f:\mathbb{R}\rightarrow\mathbb{R}$. Then given the error $\epsilon$, the expectation value of $f(J(X,T))=f(\int_{t=0}^Tg(t)X(t)\,dt)$ can be efficiently evaluated within $O(\mathcal{P}+\mathcal{G}+nl_Yl_T)$ circuit depth and $O(\frac{LP}{\epsilon}(\mathcal{P}+\mathcal{G}+nl_Yl_T))$ time complexity. 
(see proof in Appendix \ref{pf:8})
\end{thm}
\noindent Therefore our algorithms would enable the quantum-enhanced Monte Carlo method to apply to path-dependent and continuous stochastic processes, including the time-weighted expectation $\mathbb{E}[f(\frac{1}{T}\int_{t=0}^TtX(t)\,dt)]$ and the exponential decay time-weighted expectation $\mathbb{E}[f(\int_{t=0}^Te^{-\alpha t}X(t)\,dt)]$ that are usually used in mathematical finance and quantitative trading.

\subparagraph{\noindent History-Sensitive Information.} Besides the global information defined on the whole path studied above, there is another category of history-sensitive problems that plays an essential role in statistics and finance, including the first-hitting time of Brownian motion, surviving time of ruin theory, and survival analysis, to name but a few. The most apparent difference of the first-hitting problem is that the path-dependent information needs to be extracted when knowing the whole history path. Thus it can not be derived directly by a quantum walk (as illustrated in Figure \ref{first_hitting}). Formally, a first-hitting problem can be regarded as evaluating the following
\begin{equation}\label{eq:21}
K(X, T, B)=\inf{\{t:X(t)>B,T>t>0\}}.
\end{equation}
The basic idea to evaluate Eq.~\eqref{eq:21} is to consider $n$ flag qubits storing the comparison result of each piece in the sampling path. As a discontinuous jump representing extreme events or market hits always happens at the end of a piece, this method shall work much more precisely than the uniform sampling method. Formally speaking, supposing a given bound $B$, to derive the first hitting time, a bound information register and a flag register is introduced with $O(\log{B})$ and $n$ qubits, respectively. For the $j^{th}$ piece of QCTSP, the remained value is derived through a controlled subtraction whose control qubit is $F_{j-1}$ and carry-out qubit is $F_j$.
And a CNOT gate is implemented to flip the $F_j$ qubit controlled by the $F_{j-1}$ qubit. Two classes are distinguished for each piece: Before the first hit, the remaining bound is $\sum_{j'=1}^{j-1}Y_{j'} \le B$, and the flag register is $\ket*{\underbrace{11...1}_{j-1}0...0}_{flag}$. This first class can be divided into two cases: If the hit does not happen at the $j^{th}$ piece, the state of the carry-out qubit $F_j$  after the subtraction remains $\ket{0}$, and then is flipped by the CNOT gate. Hence one has: 
\begin{figure}[]
\centering
\includegraphics[width=0.8\textwidth]{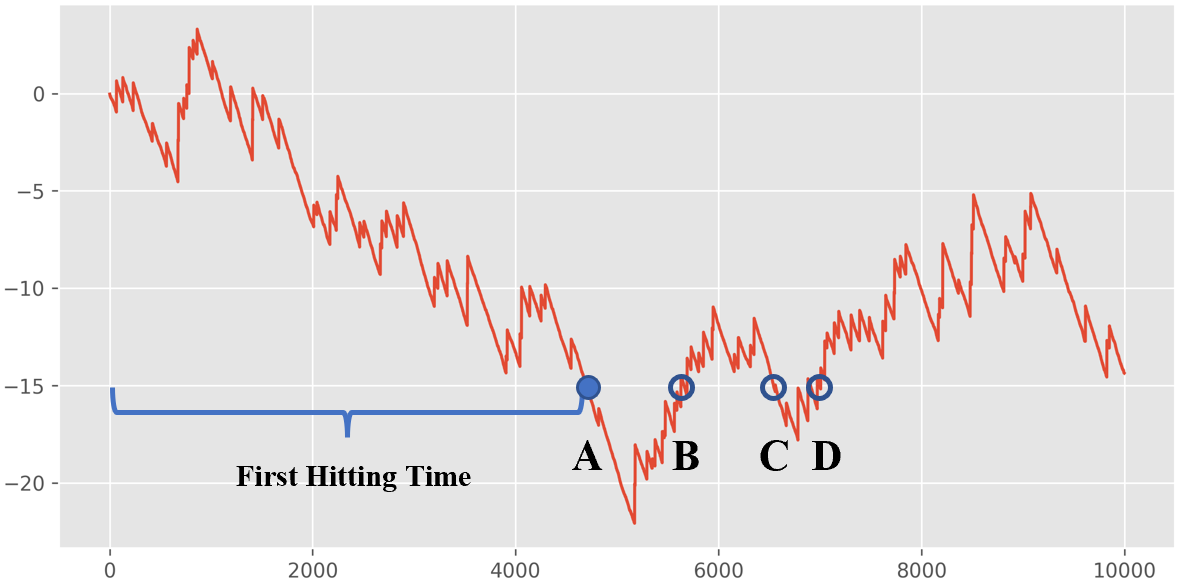}
\caption{As shown in this figure, the point \textbf{A} represents the first hitting event with the desired first hitting time, while the points \textbf{B}, \textbf{C}, \textbf{D} should not be captured. This information is determined by the knowledge of the whole past path, hence being \textit{history-sensitive}, and can not be extracted by  a quantum walk.}
\label{first_hitting}
\end{figure}
\noindent\textbf{case 1} ($\sum_{j'=1}^{j-1}Y_{j'} \le B$ and $\sum_{j'=1}^{j}Y_{j'} < B$):
\begin{equation*}
\begin{split}
&\ket*{B-\sum\nolimits_{j'=1}^{j-1}Y_{j'}}_{bound}\ket*{\underbrace{11...1}_{j-1}0...0}_{flag}\\
\xrightarrow[\text{subtractor}]{\text{controlled }}&\ket*{B-\sum\nolimits_{j'=1}^{j}Y_{j'}}_{bound}\ket*{\underbrace{11...1}_{j-1}0...0}_{flag}\\
\xrightarrow{CNOT}&\ket*{B-\sum\nolimits_{j'=1}^{j}Y_{j'}}_{bound}\ket*{\underbrace{11...1}_{j}0...0}_{flag}.\\
\end{split}
\end{equation*}
If the hit happens at the $j^{th}$ piece, the state of the carry-out qubit $F_j$  after the subtraction turns to be $\ket{1}$, and then is flipped by the CNOT gate. Hence one has: 

\noindent\textbf{case 2} ($\sum_{j'=1}^{j-1}Y_{j'} \le B$ and $\sum_{j'=1}^{j}Y_{j'} > B$):
\begin{equation*}
\begin{split}
&\ket*{B-\sum\nolimits_{j'=1}^{j-1}Y_{j'}}_{bound}\ket*{\underbrace{11...1}_{j-1}0...0}_{flag}\\
\xrightarrow[\text{subtractor}]{\text{controlled }}&\ket*{B-\sum\nolimits_{j'=1}^{j}Y_{j'}}_{bound}\ket*{\underbrace{11...1}_{j}0...0}_{flag}\\
\xrightarrow{CNOT}&\ket*{B-\sum\nolimits_{j'=1}^{j}Y_{j'}}_{bound}\ket*{\underbrace{11...1}_{j-1}0...0}_{flag}\\
\end{split}
\end{equation*}
For the second class where the current piece is after the first hit $\sum_{j'=1}^{j''}Y_{j'} > B$, the flag register is $\ket*{\underbrace{11...1}_{j''}0...0}_{flag}$, and the subtraction and CNOT gates will not be implemented, and thus the result is:

\noindent\textbf{case 3} ($\sum_{j'=1}^{j''}Y_{j'} > B$ for some $j''< j$):
\begin{equation*}
\begin{split}
&\ket*{B-\sum\nolimits_{j'=1}^{j''}Y_{j'}}_{bound}\ket*{\underbrace{11...1}_{j''-1}0...0}_{flag}\\
\xrightarrow[\text{subtractor}]{\text{controlled }}&\ket*{B-\sum\nolimits_{j'=1}^{j''}Y_{j'}}_{bound}\ket*{\underbrace{11...1}_{j''-1}0...0}_{flag}\\
\xrightarrow{CNOT}&\ket*{B-\sum\nolimits_{j'=1}^{j''}Y_{j'}}_{bound}\ket*{\underbrace{11...1}_{j''-1}0...0}_{flag}\\
\end{split}
\end{equation*}
Leaving the theoretical analysis as future work, examples of option pricing and ruin theory will be given in the next section as proof of principle.

\section{APPLICATION}\label{sec:4}

In this section, two applications of QCTSP in different financial scenes are given to illustrate the wide range of applicability of QCTSP. 
In subsection \ref{sec:merton}, the first application is focused on the option pricing problem wherein the underlying stock price is no longer continuous Brownian motion and the original \textit{Black-Scholes} formula fails. By simulation of QCTSP, option pricing is solved in a different way from previous works and is consistent with the more practical \textit{Merton Jump Diffusion} formula. 
In subsection \ref{sec:ruin}, the problem of computing the ruin probability under the \textit{Collective Risk Model} is studied, introducing quantum computation into insurance mathematics.

\subsection{Option pricing in Merton Jump Diffusion Model}\label{sec:merton}

Since first being introduced into the field of financial engineering, CTSP has been proven to be a powerful tool for financial derivatives pricing \cite{merton1998applications,kwok2008mathematical}. Among those financial derivatives, the European call option is one basic instrument that gives someone the right to buy an underlying stock $S_T$ at a given strike price $K$ and maturity time $T$. The famous \textit{Black-Scholes} model is proposed to evaluate a European option \cite{black2019pricing}, and the simulation for \textit{Black-Scholes} type option pricing has been implemented on quantum computers \cite{rebentrost2018quantum, stamatopoulos2020option, martin2019towards, blank2021quantum}. However, the assumption of a constant variance log-normal distribution in the original article \cite{black2019pricing} turns out to be less realistic as the empirical studies of discontinuous returns of stock are ignored. Hence the \textit{Merton Jump Diffusion Model} was proposed to incorporate more realistic assumptions from Merton's work \cite{merton1976option}.

In the \textit{Merton Jump Diffusion Model}, the stock price $S_T$ is assumed to satisfy the following stochastic differential equation:
\begin{equation} \label{eq:22}
\begin{split}
\ln{S_T}=&\ln{S_0}+\int_{t=0}^T(r-\frac{\sigma^2}{2}-\lambda(m+\frac{v^2}{2}))\,dt\\
&+\int_{t=0}^T\sigma\,{dW(t)}\\
&+\sum_{j=1}^{Poi(T)}(J_j-1),
\end{split}
\end{equation}
where $S_0$ is the current stock price, $T$ is the maturity time in years, $r$ is the annual risk-less interest rate, $\sigma$ is the annual volatility, $\,{dW(t)}$ is the Weiner process, $Poi(T)$ is the \textit{Poisson point process}, and $J_j$ is the $j^{th}$ discontinuous jump follows a log-normally distribution. The last term on the right hand side represents the discontinuous price movements caused by such as acquisitions, mergers, corporate scandals, and fat-finger errors. Merton derives a closed form solution $MJD(S,K,T,r,\sigma,m,v,\lambda)$ to Eq.~\eqref{eq:22} as an infinite summation of \textit{Black-Scholes} formula $BS(S,K,T,r_j,\sigma_j)$ conditional on the number of the jumps and the underlying distributions of each jump:
\begin{equation}\label{eq:merton2}
MJD(S,K,T,r,\sigma,m,v,\lambda)=\sum_{j=0}^\infty BS(S,K,T,r_j,\sigma_j),
\end{equation}
where $m$ is the mean of the jump size, $v$ is the volatility of the jump size, $\lambda$ is the intensity of the underlying \textit{Poisson point process}, and  $r_j=r-\lambda(m-1)+\frac{j\ln{m}}{T}$ and $\sigma_j=\sqrt{\sigma^2+j\frac{v^2}{T}}$ are the corresponding interest rate and volatility, respectively. 
Despite being in line with empirical studies of market returns, the formula Eq.~\eqref{eq:merton2} takes the form of an infinite summation, and thus an alternative numerical method of Monte Carlo can be employed to compute the desired option price. Instead of a Brownian motion in the  \textit{Black-Scholes} model, one has to simulate, in the  \textit{Merton Jump Diffusion Model}, the more complicated case of a L\'{e}vy process where our method can be applied to make a quadratic speed-up against the classical Monte Carlo simulation. More precisely, for a European type call option given the underlying stock price $S_T$, its price  can be defined as:
\begin{equation}
\begin{split}
    price=&\max\{S_0e^{D+B_T+\sum_{j=1}^{Poi(T)}J_j}-K,0\}\\
    =&S_0e^{\max\{D+B_T+\sum_{j=1}^{Poi(T)}J_j, \ln{\frac{K}{S_0}}\}} - K,
\end{split}
\end{equation}
where $D=(r-\frac{\sigma^2}{2}-\lambda(m+\frac{v^2}{2}))T$ is the adjusted drift term on the purpose of risk neutral preferences, $B_T$ represents the Brownian motion term and follows a normal distribution $N(0, \sigma\sqrt{T})$, and $J_j$ represents the $j^{th}$ discontinuous jump and follows a normal distribution $N(m,v)$.

\begin{figure}[h!]
\centering
\includegraphics[width=0.8\textwidth]{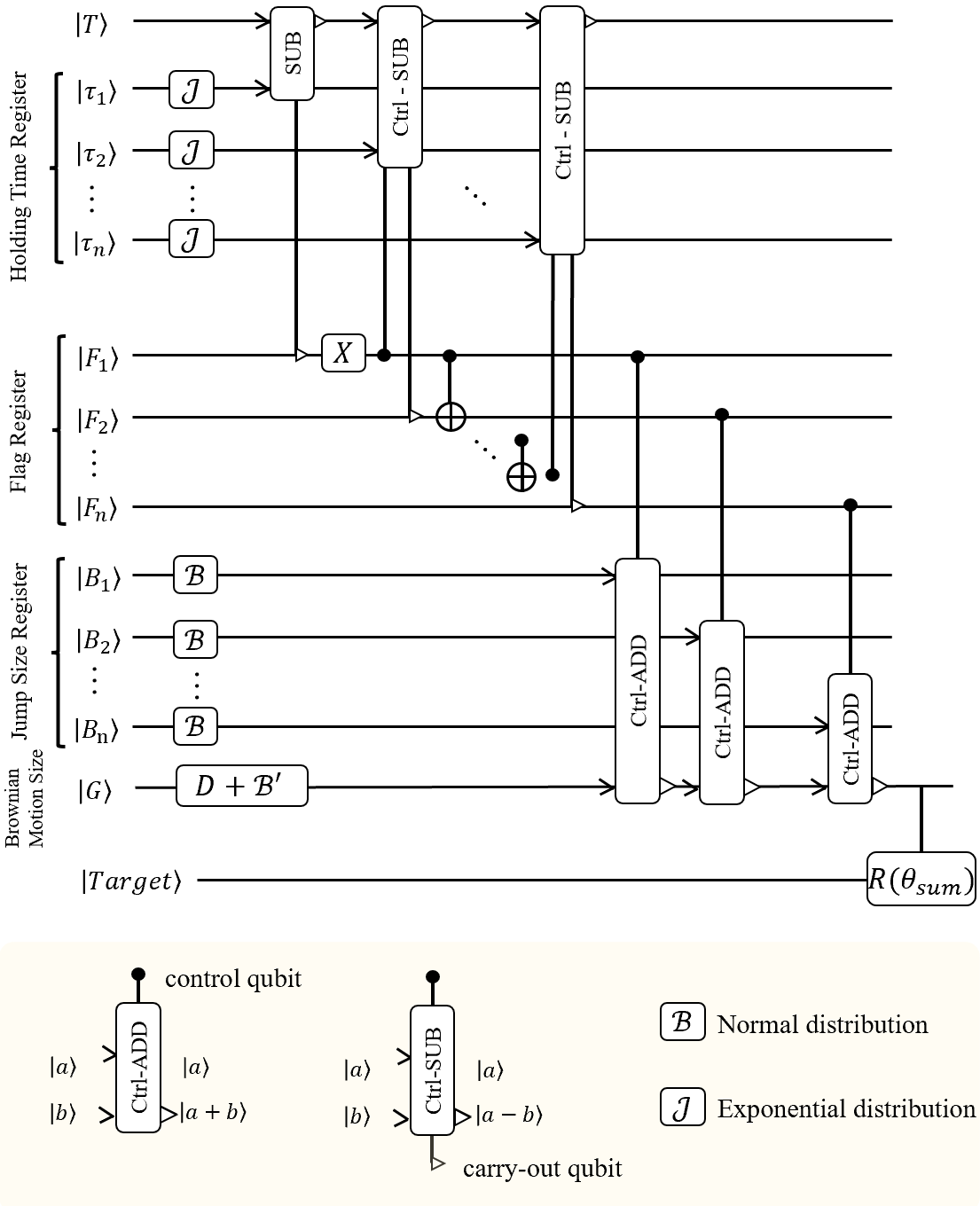}
\caption{\textbf{QCTSP Circuit of the Simulation of Merton Jump Diffusion Model.} In this figure, the circuit of the simulation of the European call option price in \textit{Merton Jump Diffusion Model} is presented. 
The holding time $\tau_j$ is prepared via parallel subcircuits $\mathcal{J}=Exp(\lambda T)$, and the jump size $B_j$ is prepared via  subcircuits $\mathcal{B}=Norm(m,v)$. 
The drift term, together with the  geometric Brownian motion, is prepared on $\ket{G}$ via the  subcircuit $\mathcal{D+B'}=Norm((r-\frac{\sigma^2}{2}-\lambda(m+\frac{v^2}{2}))T,\sigma\sqrt{T})$.  
Subtractor and controlled-subtractors, together with X gate and CNOT gates, are implemented to justify the number of jumps with output on the flag register $\ket{F_j}$. Adders controlled by the flag registers are employed to derive the final value of the stock price, followed by a rotation gate on the target qubit controlled by the signed summation on register $\ket{\theta_{sum}}$.}
\label{MertonJumpCircuit}
\end{figure}

 This L\'{e}vy process can be efficiently prepared as mentioned in subsection \ref{sec:levy}, and then the valuation procedure is implemented via the method given in section \ref{sec:extract}.  More specifically, as illustrated in Figure \ref{MertonJumpCircuit}, the quantum circuit consists of six quantum registers: the maturity time register $\ket{T}$, the holding time register $\ket{\tau_j}$, the flag register $\ket{F_j}$, the jump size register $\ket{B_j}$, the Brownian motion size register $\ket{G}$, and the target qubit $\ket{Target}$. The holding time $\tau_j$ can be prepared via parallel exponential distribution subcircuits $\mathcal{J}=Exp(\lambda T)$, and the jump size $B_j$ can be prepared via parallel normal distribution subcircuits $\mathcal{B}=Norm(m,v)$. The drift term together with the  geometric Brownian motion is prepared on $\ket{G}$ via a normal distribution subcircuit $\mathcal{D+B'}=Norm((r-\frac{\sigma^2}{2}-\lambda(m+\frac{v^2}{2}))T,\sigma\sqrt{T})$.  Subtractor and controlled-subtractors, together with X gate and CNOT gates, are implemented to justify the number of jumps that will be considered, and the result is output on the flag register $\ket{F_j}$. Adders controlled by the corresponding flag registers are employed to derive the final value of the stock price. A rotation gate is then implemented on the target qubit, controlled by the signed summation on register $\ket{\theta_{sum}}$. Then the  final value of the option price can be evaluated via measurement or an amplitude estimation. 
\begin{figure}[h]
\centering
\includegraphics[width=0.85\textwidth]{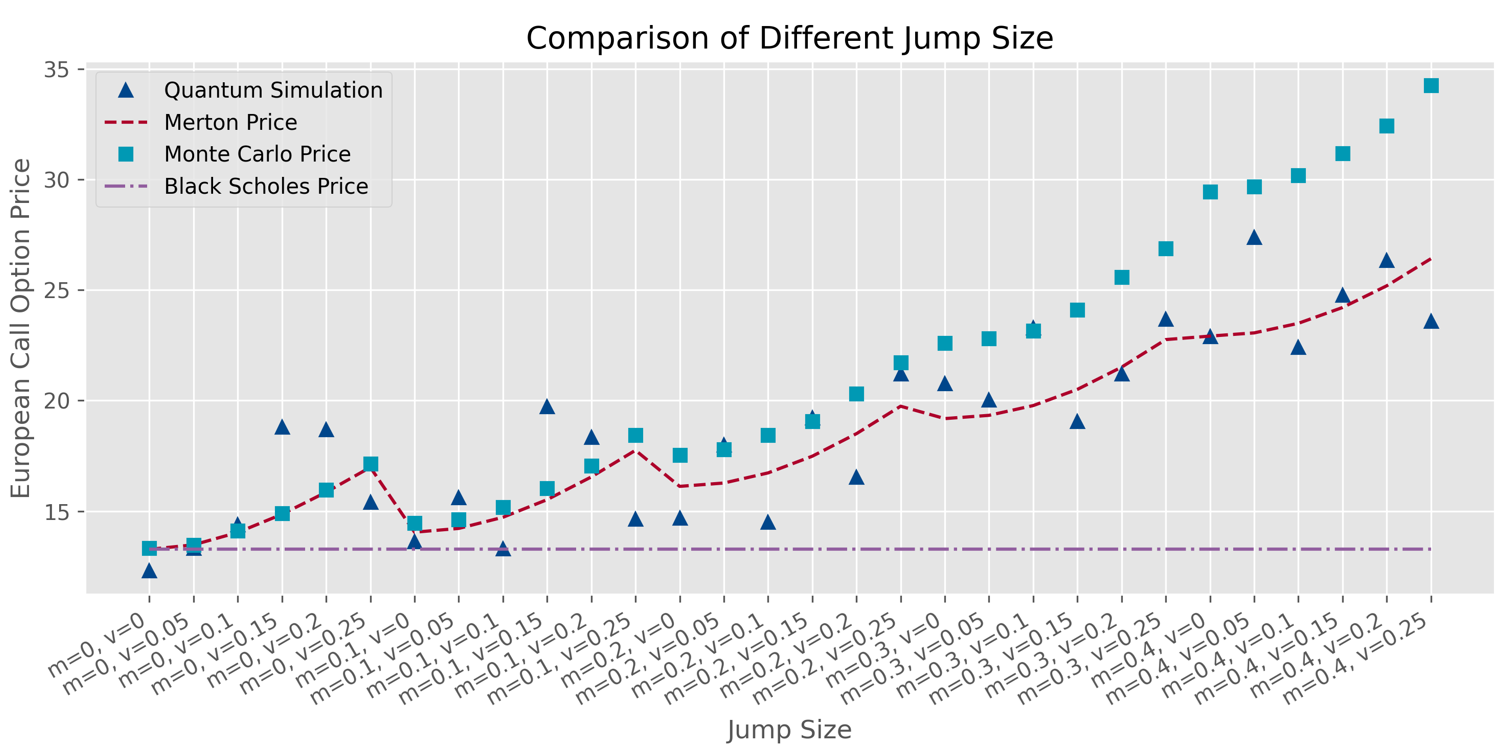}
\caption{\textbf{Comparison of QCTSP Simulation of European Call Option Price in Merton Jump Diffusion Model with Different Jump Size.} The simulation of European call option price in \textit{Merton Jump Diffusion Model} is presented in this figure. The parameters for the underlying stock price are set as $S_0=100$, $K=100$, $r=0.1$, $\sigma=0.02$, $T=1$, and  $\Delta t=\frac{1}{30}$, and the range of $m$ and $v$ are $0, 0.1, 0.2, 0.3, 0.4$ and $0, 0.05, 0.1, 0.15, 0.2, 0.25$, respectively.  As shown in the figure, the simulated prices are consistent with the Merton formula Eq.~\eqref{eq:merton2}, and the option price deviates from the \textit{Black Scholes} formula as the jump size tends to be large.}
\label{MertonCompJumpSize}
\end{figure}
\begin{figure}[hp!]
\centering
\includegraphics[width=0.68\textwidth]{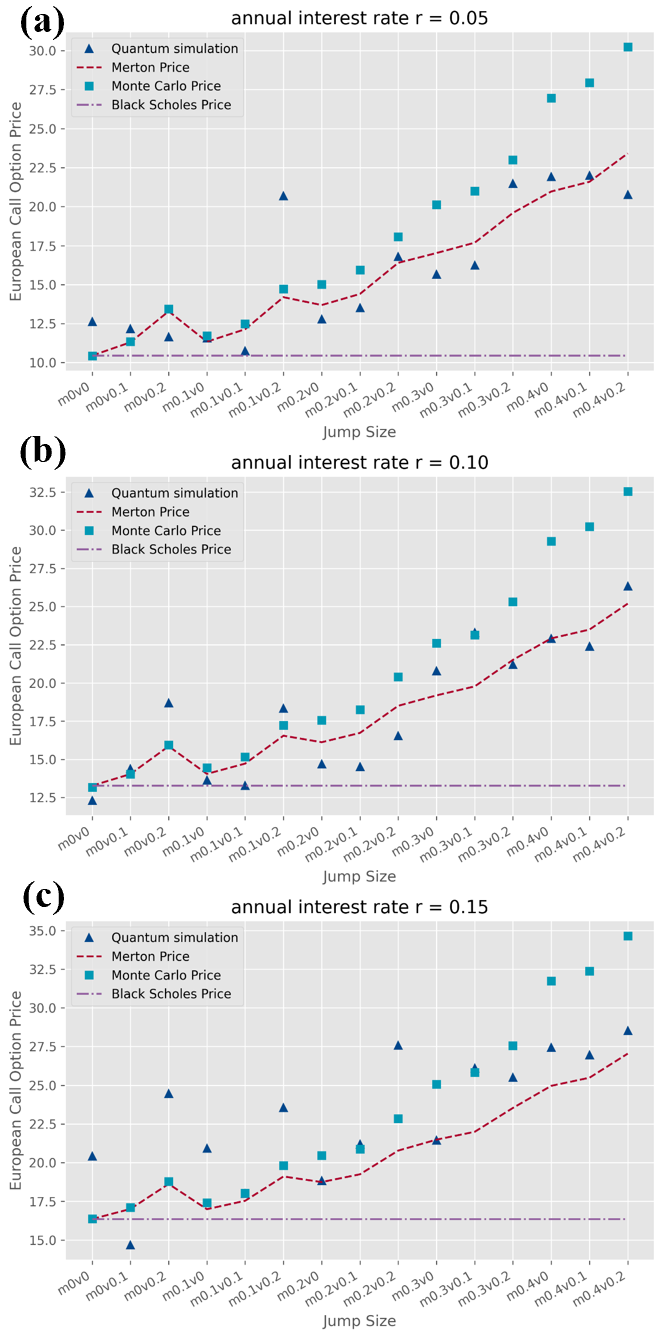}
\caption{\textbf{Comparison of QCTSP Simulation of European Call Option Price in Merton Jump Diffusion Model under Different Annual Interest.} The simulation of European call option price in \textit{Merton Jump Diffusion Model} is presented in this figure. The underlying stock prices of annual interest rate $r=0.05, 0.1, 0.15$ are given with $S_0=100$, $K=100$, $\sigma=0.02$, $T=1$, and  $\Delta t=\frac{1}{30}$, and the range of $m$ and $v$ are $0, 0.1, 0.2, 0.3, 0.4$ and $0, 0.1, 0.2$, respectively.  As shown in the figure, the simulation results are consistent with the Merton formula Eq.~\eqref{eq:merton2} while the interest rate, the average jump size, and the volatility of jump size vary, showing robustness with wide ranges of parameters.}
\label{MertonCompInterestRate}
\end{figure}
 The corresponding quantum simulation result can be found in Figure \ref{MertonCompJumpSize} and Figure \ref{MertonCompInterestRate}. In Figure \ref{MertonCompJumpSize}, the underlying stock price is given with $S_0=100$, $K=100$, $r=0.1$, $\sigma=0.02$, $T=1$, and  $\Delta t=\frac{1}{30}$, and the range of $m$ and $v$ are $0, 0.1, 0.2, 0.3, 0.4$ and $0, 0.05, 0.15, 0.2, 0.25$, respectively. Due to the restriction on the qubit number, the simulation is divided into two steps: $100$ groups of random rotation angles corresponding to the Brownian motion size. The jump sizes are generated by a classical random generator and then implemented on the target qubit. $1024$ shots of QCTSP simulated paths are repeated for each group. As shown in Figure \ref{MertonCompJumpSize}, the larger the jump size (characterized by $m$ and $v$) is, the farther the \textit{Merton Jump Diffusion} value is away from the original \textit{Black-Scholes} price as a consequence of the discontinuous jumps. Moreover, the QCTSP simulation result is consistent with the $40$ terms truncated \textit{Merton Jump Diffusion} formula (Eq.~\eqref{eq:merton2}) as well as the classical Monte Carlo simulation, characterizing the property of L\'{e}vy discontinuous jump path well. The robustness of the QCTSP method is illustrated in Figure \ref{MertonCompInterestRate}, where different annual interest rates $r=0.05, 0.1, 0.15$ are given in the three subfigures with varying jump sizes. The quantum simulation results are consistent with the Merton Formula for a wide range of different parameters.
 
\subsection{Ruin probability in Collective Risk Model}\label{sec:ruin}

Since being a fundamental means to model the stochastic world, there is no surprise that the QCTSP framework developed in this work has great potential power applied to various fields, especially ruin theory. The ruin theory plays a central role in insurance mathematics, high-frequency trading's market micro-structure theory, and option pricing \cite{gerber1998time, garman1976market, madhavan2000market, gerber1999ruin}. In ruin theory, the risk of an insurance company is assumed to be caused by random claims arriving at time $T_j=\sum_{j'=1}^j \tau_{j'}$, wherein the $j^{th}$ claim size $\xi_j$ and inter-claim time $\tau_j$ are both assumed to follow some independent and identically distributions. Hence the aggregate asset of the insurance company is a continuous time stochastic process
\begin{equation}
X(t)=u+ct-\sum_{j=1}^{Poi(t)} \xi_j,
\end{equation}
where the initial surplus is $X(0)=u$, and the premiums are received at a constant rate $c$. In the \textit{Collective Risk Model}, also known as the Cram\'{e}r Lundberg model, the claim number process $Poi(t)$ is assumed to be a Poisson process with intensity $\lambda$, and the underlying distribution of $\tau_j$ is an exponential distribution. In the Sparre Andersen model, $Poi(t)$ can be extended to a renewal process with the arbitrary underlying distribution. Despite the detailed differences between them, both of the two stochastic processes $X(t)$ in these models are L\'{e}vy processes with drift term $ct$,  and hence can be prepared easily via our method. 
\begin{equation}
Y(t)=u+ct-X(t)=\sum_{j=1}^{Poi(t)} \xi_j
\end{equation}
To go one step further, one has a great variety of ruin-related quantities fall into the category of the expected discounted penalty function, also known as the time value of ruin. And those can be easily derived through a quantum-enhanced Monte Carlo modified by us. More precisely, following the notation of Gerber and Shiu \cite{gerber1998time}, the time value of ruin is defined as 
\begin{equation} \label{eq:20}
\phi{(u)}=\mathbb{E}^{X_t}[e^{-\delta\tau} w(X(\tau-), X(\tau))\mathbb{I}_{\tau<\infty}|X(0)=u],
\end{equation}
where $\tau{(u)}=\inf{\{t:U(t)<0|U(0)=u\}}$ denotes the time of ruin, and $X(\tau-)=u+c\tau-\sum_{j=1}^{N(\tau)-1} X_j$ and $X(\tau)=u+c\tau-\sum_{j=1}^{N(\tau)} X_j$ are the surplus prior to ruin and the deficit at ruin, respectively. The expectation is taken over the probability distribution of the ruin samples, taking an interest discounting factor $e^{-\delta\tau}$ into consideration. The ultimate ruin probability $\psi(u)=\mathbb{P}[X(\tau)<0|X(0)=u,\tau<\infty]$ is exactly a special case of Eq.~\eqref{eq:20} given $\delta=0$ and $w(x,y)=1$. Since ruin always happens after a claim, 
\begin{align*}
\psi(u)=&\mathbb{P}[X(\tau)<0|X(0)=u,\tau<\infty]\\
=&\mathbb{E}^{X_t}[\mathbb{I}_{\tau<\infty}|X(0)=u]\\
=&\mathbb{E}^{X_t}[\mathbb{I}_{u+\sum_{j=1}^n (\tau_j-\xi_j)<0}|n<\infty]\\
=&\mathbb{E}^{Y_t}[\mathbb{I}_{Y(t)<u+ct}],
\end{align*}
where the indicator function $\mathbb{I}_{Y(t)<u+ct}$ can be easily derived through a sequence of controlled adder and modified subtractor on the prepared states. The detailed construction of circuits and corresponding simulation result are given in Figure \ref{circuit_ruin_probability} and Figure \ref{sim_ruin_probability}.
\begin{figure}[h!]
\centering
\includegraphics[width=0.65\textwidth]{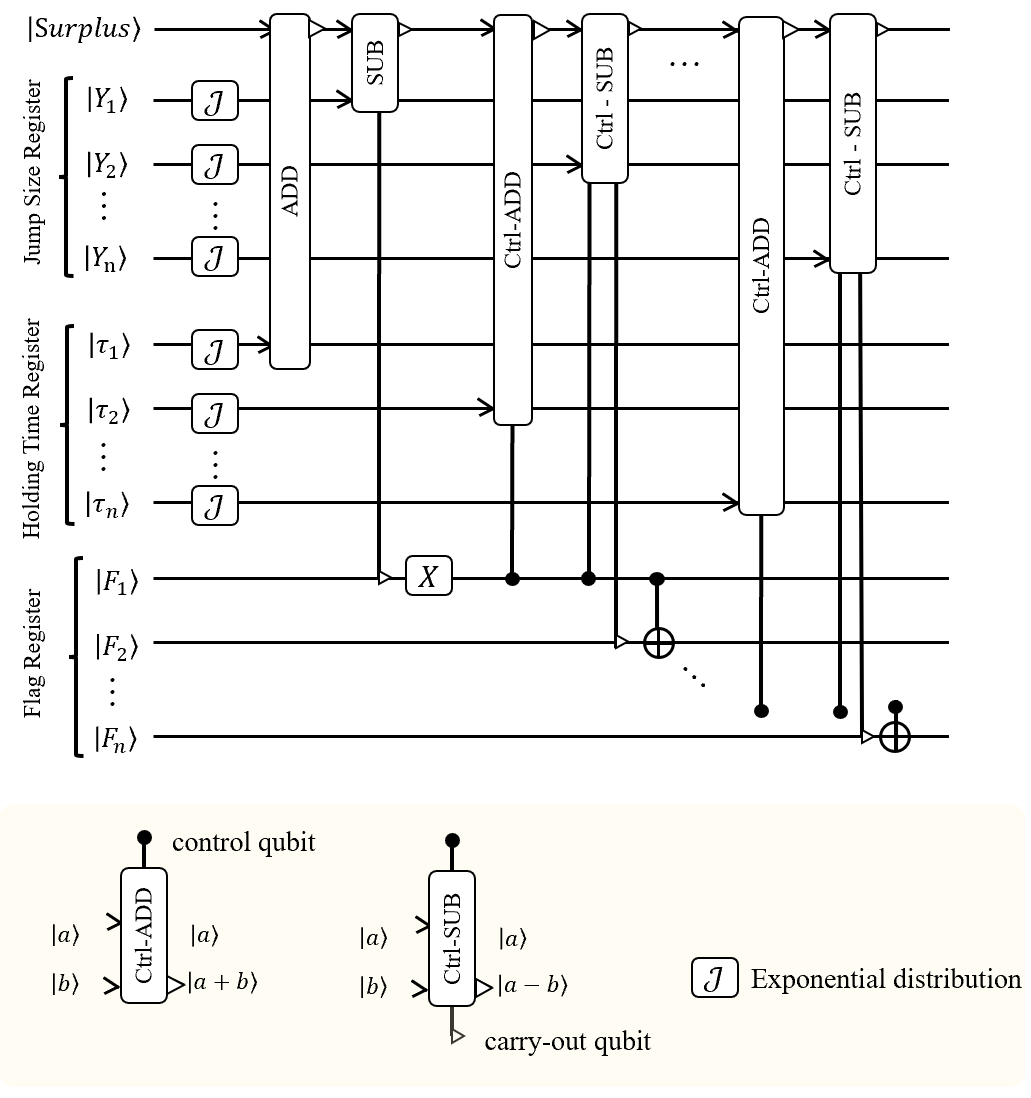}
\caption{\textbf{Quantum Circuit of Computing Ruin Probability.} The ultimate probability of the \textit{Collective Risk Model} is simulated via this quantum circuit. The receiving rate is assumed to be $1$ so that the premium can be summed directly, removing the qubits requirement and the circuit depth from function $cT$. For each piece of QCTSP, an adder is followed by a controlled-subtracor whose carry-out qubit is on the flag qubit together with a controlled flip, as discussed in subsection \ref{sec:extract}.}
\label{circuit_ruin_probability}
\end{figure}

\begin{figure}[h!]
\centering
\includegraphics[width=0.85\textwidth]{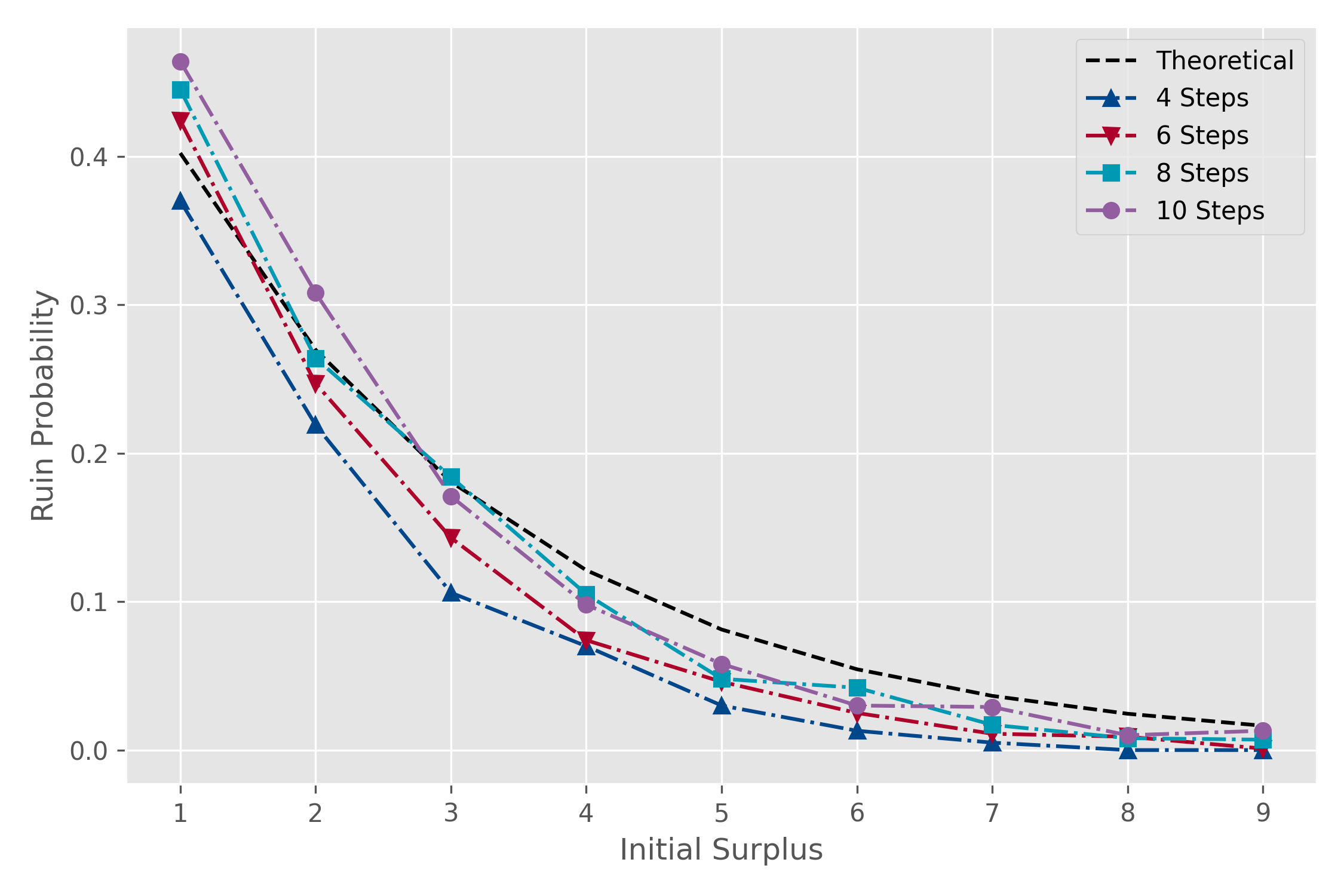}
\caption{\textbf{Simulation of Ultimate Ruin Probability.} 
The ultimate probability of the \textit{Collective Risk Model} is simulated for varying initial surplus with sample paths of different lengths. 
The intensities of inter-claim time and claim size are $\lambda_t=0.6$, and $\lambda_\xi=1$, respectively. 
The sample precision is $\epsilon=0.001$, and the premiums' receiving rate is $c=1$. 
The shot number is set to be 1000 for each simulation point. 
As depicted in the figure, the simulated ruin probabilities for varying surpluses align well with the theoretical benchmark. 
Furthermore, we notice that as the length of the sample paths increases, the deviation from the theoretical curve diminishes, indicating improved accuracy in our estimations.
It should be noticed that the ultimate ruin probability represents the probability of ruin after an infinitely long time, and it can be approximated through simulation for large values of time $T$. 
The advantage of the QCTSP method lies in its ability to simulate over much longer time intervals, allowing for a remarkably close approximation of this ultimate ruin probability.}
\label{sim_ruin_probability}\end{figure}

\section{DISCUSSION}\label{sec:5}

In this study, we have focused on the efficient state preparation and quantum-enhanced analysis of continuous time stochastic processes (QCTSP). We have successfully established a robust procedure for QCTSP preparation, allowing the state to be sensitized to discontinuous jumps, which are crucial in modeling extreme market emergencies. Notably, we have achieved significant reductions in both qubit number and circuit depth, particularly concerning the key parameter of holding time.

Furthermore, we are excited to present our novel Monte Carlo simulation method tailored for QCTSP. This groundbreaking approach has led to a remarkable quadratic speed-up compared to the classical counterpart, the continuous time stochastic process (CTSP), making QCTSP a formidable tool in studying the continuous time stochastic world.

A significant contribution of our work lies in the development of techniques for extracting \textit{weighted integral} and \textit{history-sensitive information}. These methods are of paramount importance in various fields, such as quantitative trading, time series analysis, and actuarial mathematics. By eliminating the need for additional $n$ copies of random paths and relaxing the implied restriction of the independent and identically distributed (I.I.D.) condition, our techniques open new avenues for accurate analysis and computation.

In the context of practical applications, we have successfully demonstrated two scenarios: European option pricing under the \textit{Merton Jump Diffusion Model} and the evaluation of ruin probability in the \textit{Collective Risk Model}. These applications serve as compelling evidence of QCTSP's potential and versatility in addressing real-world challenges.

There are four main reasons why QCTSP stands as a suitable candidate for practical applications in the noisy intermediate-scale quantum era:

Firstly, QCTSP does not impose input restrictions and can be implemented without assuming the need for a quantum oracle or quantum Random Access Memory (qRAM). This property enhances its versatility and applicability in various quantum algorithms, including quantum machine learning, thereby breaking through the input bottleneck often encountered in quantum computations.

Secondly, QCTSP allows for parallel preparation, thanks to a communicative equivalence class decomposition of the sample space. This decomposition results in a reduction in the number of required qubits, contributing to more efficient and resource-effective implementations.

Thirdly, our QCTSP preparation method imposes less requirements on the topological structure of the quantum processor. In many cases, the connectivity requirement, which concerns memory length, can be set at a low level. Consequently, most qubits only need to be entangled with their immediate neighbors, easing the burden of implementation and enhancing scalability.

Lastly, considering the fundamental and critical role that Continuous Time Stochastic Process (CTSP) plays in stochastic analysis, QCTSP holds immense potential for theoretical generalizability and applied flexibility. Its relevance in a wide range of stochastic processes further solidifies its status as a promising candidate for practical applications.

QCTSP exhibits remarkable advantages in terms of input flexibility, parallel preparation, compatibility with quantum algorithms, and potential for theoretical and applied advancements. These features position QCTSP as a promising tool for practical utilization in the emerging era of noisy intermediate-scale quantum computing.

{While our work showcases the promising potential of QCTSP, it is essential to acknowledge certain limitations that warrant further investigation and exploration. We outline some of these limitations below:}

{Firstly, while QCTSP is believed to be hardware-efficient due to its low requirements on the topological structure, a more detailed and quantitative analysis is lacking in this article. We acknowledge this as an avenue for future research and plan to address it through experimental demonstrations.}

{Secondly, we have not explored the preparation of many other CTSPs with financial relevance, such as Hawkes processes, which may exhibit rich structures and evolutionary laws that could potentially lead to further quantum speed-up. These CTSPs warrant further investigation to understand their quantum computational implications.}

{Thirdly, while we have demonstrated the applications of QCTSP in European option pricing and ruin probability evaluation, many other relevant applications in fields such as high-frequency trading, actuarial science, and option pricing have not been explored in this work. We recognize the significance of studying these facets and plan to investigate them in future research while maintaining the thematic focus of this current work.}

Furthermore, future research should focus on developing more techniques for the efficient extraction of path-dependent information. This aspect is crucial for enhancing the versatility and practicality of QCTSP in various financial analyses.

{Additionally, there are intriguing topics related to the intrinsic relation of continuous quantum walk and finance, as discussed in \cite{choustova2009quantum, shikano2013discrete, chang2023preparing}. While discrete and continuous quantum walk methods have shown promise as powerful tools for financial data state preparation and stochastic simulation in finance, we have not delved into this area in the current work. One of the challenges lies in finding suitable ways to store the history information of continuous quantum walks. We leave the exploration of this relationship as a promising avenue for future research.}

In conclusion, while our work has made significant strides in the efficient state preparation and quantum-enhanced analysis of continuous time stochastic processes, there are various areas that warrant further investigation and study. We are committed to exploring these limitations and expanding the scope of our research in future endeavors.

\section*{ACKNOWLEDGEMENT}

This work was supported by the National Natural Science Foundation of China (Grants No. 12034018), and Innovation Program for Quantum Science and Technology No. 2021ZD0302300.

\bibliographystyle{quantum}
\bibliography{EXBIB.bib}

\onecolumn\newpage
\appendix

\section{Basic Circuits for Arithmetic and Distribution Preparation}\label{sec:6}
Some basic arithmetic circuits and statistics distribution preparation circuits are given in this appendix. 

\subsection{Modified Quantum Subtractor}\label{app:arithmetic}
Although adder and subtractor circuits have been studied in many works(see \cite{cuccaro2004new} for reference), some modification is needed in our work as we wish to map $\ket{a}\ket{b}$ to $\ket{a}\ket{a-b}$, i.e., store the result in the second register while leave the first register unchanged. The basic idea is as follows: since $(a'+b)'=a-b$, X gates are introduced to calculate the complement of a as $\ket{a'}\ket{b}$. Then a ripple-carry addition circuit is implemented to output the summation result on the first register $\ket{a'+b}\ket{b}$. Finally, another sequence of X gates are put on the first register $\ket{a-b}{b}$.

\subsection{Quantum Multiplier for Amplitude Estimation and Quantum-Enhanced Monte Carlo}
In this subsection, a quantum multiplier for quantum estimation, and hence quantum Monte Carlo is given as follows:
\begin{lem}\label{thm:10}
Suppose that $\ket{a}$ and $\ket{b}$ are $l_a$-bit and $l_b$-bit strings, respectively, and $\ket{e^{iM\theta_0}}$ is an analog-encoded qubit, then the multiplication of $\ket{a}$ and $\ket{b}$ can be added to $\ket{e^{iM\theta_0}}$  as $\ket{a}\ket{b}\ket{e^{iM\theta_0}}\rightarrow\ket{a}\ket{b}\ket{e^{i(M+ab)\theta_0}}$, within $l_al_b$ controlled-rotation gates.
\end{lem}

\begin{proof}
The proof is directly: Given two bit strings $\ket{a} = \ket{x_{l_a}...x_2x_1}$ and $\ket{b} = \ket{y_{l_b}...y_2y_1}$, their product can be computed as
\begin{align*}
ab=&(\sum_{i=1}^{l_a}x_i2^{i-1})(\sum_{j=1}^{l_b}y_j2^{j-1})\\
=&\sum_{i=1}^{l_a}\sum_{j=1}^{l_b}x_iy_j2^{i+j-2}.
\end{align*}
Hence the term $x_iy_j2^{i+j-2}$ can be implemented by a rotation gate on the target qubit $\ket{e^{iM\theta_0}}$ controlled by the $i^{th}$ and $j^{th}$ qubits of register $\ket{a}$ and $\ket{b}$, where the rotation angle $\theta_{i,j}=2^{i+j-2}\theta_0$ is determined by the index $i,j$ only. The final state is:
\begin{align*}
\ket{e^{iM\theta_0}}\rightarrow&\ket{e^{i(M\theta_0+\sum_{i=1}^{l_a}\sum_{j=1}^{l_b}x_iy_j2^{i+j-2}\theta_0)}}\\
=&\ket{e^{i(M+ab)\theta_0}}.
\end{align*}
If $\ket{a}$ and $\ket{b}$ are signed integers, the rotation gates of two different directions should be controlled by the two sign qubits as well. The total number of controlled rotation gates is $l_al_b$ as claimed.
\end{proof}

\begin{figure}[hp]
\includegraphics[width=0.85\textwidth]{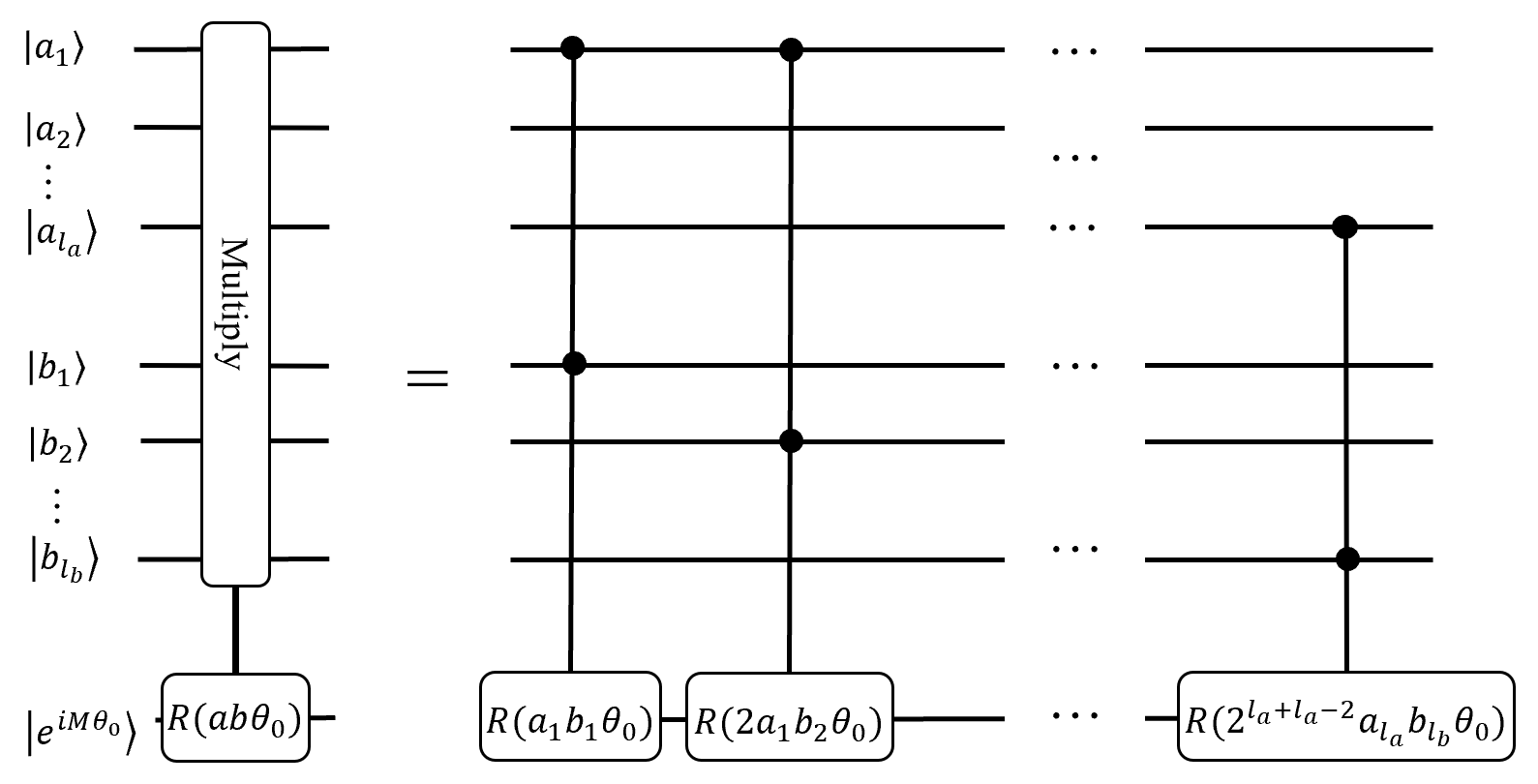}
\caption{\textbf{Modified Subtractor Circuit.} In this figure, a subtractor circuit of two 3-qubits integers that maps $\ket{a}\ket{b}$ to $\ket{a-b}\ket{b}$ is given. The basic idea is from \cite{cuccaro2004new} with some modification so that the result can be stored in the first register leaving the second register unchanged. }
\label{multiplier_circuit}\end{figure}

\subsection{State preparation for Statistical Distributions}
The exponential distribution is implemented by parallel rotation gates as illustrated in Figure  \ref{exponential_distribution_circuit}. Detailed computation of rotation angles is given in Appendix \ref{pf:1}.
The Erlang distribution is a summation of several exponential distributions and hence can be prepared via copies of exponential distribution preparation subcircuits together with a sequence of adder operators as shown in Figure  \ref{eldc}.
 Corresponding simulation results are given in Figure  \ref{exponential_distribution_simulation} and \ref{elds}, respectively.
 
\begin{figure}
\subfigure[Exponential Distribution Circuit]{
\includegraphics[width=0.44\textwidth]{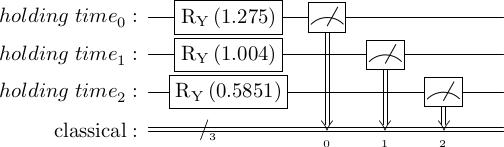}
\label{exponential_distribution_circuit}}
\subfigure[Exponential Distribution Approximation]{
\includegraphics[width=0.44\textwidth]{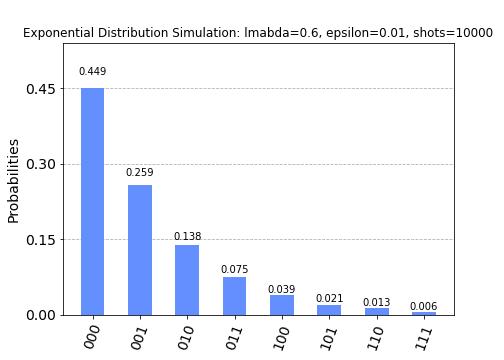}
\label{exponential_distribution_simulation}}
\subfigure[Erlang Distribution Circuit]{
\includegraphics[width=0.44\textwidth]{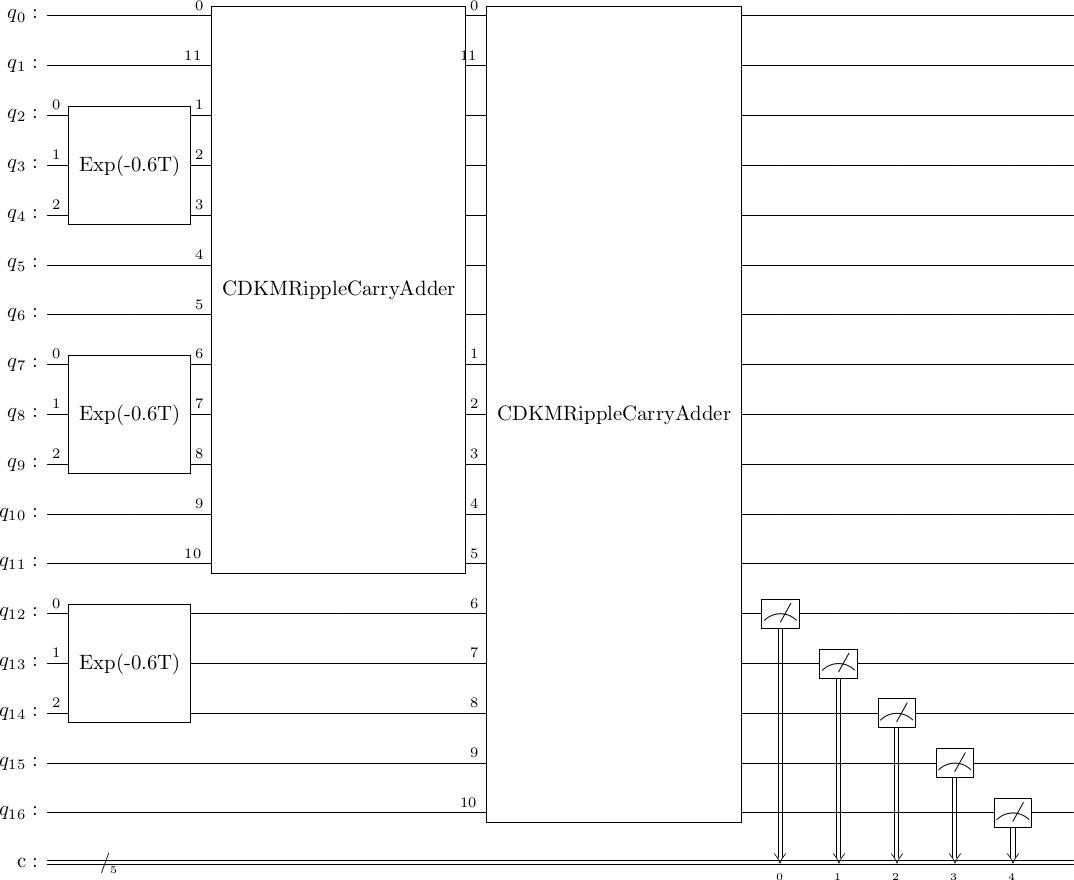}
\label{eldc}}
\subfigure[Erlang Distribution Approximation]{
\includegraphics[width=0.44\textwidth]{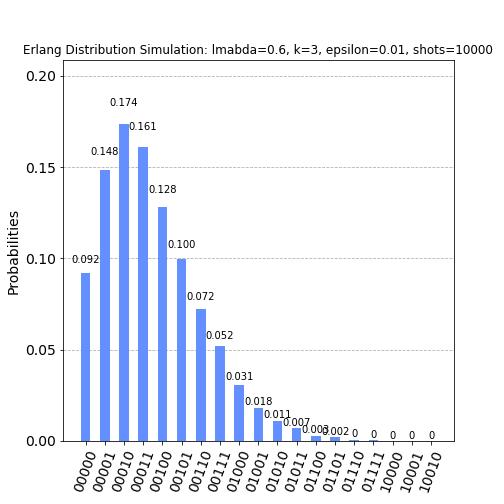}
\label{elds}}
\caption{Preparation of Exponential and Erlang Distribution}
\end{figure}

\section{Proofs for QCTSP State Preparation Problem}\label{sec:7}

\subsection{Proof of Theorem \ref{thm:1}\label{pf:1}}
\begin{proof}
By the definition of memory-less process, one has
$$\mathbb{P}[\tau_l>s]=\mathbb{P}[\tau_l>t+s|\tau_l>t]=\frac{\mathbb{P}[\tau_l>t+s]}{\mathbb{P}[\tau_l>t]},$$
and hence the c.d.f. $f(t)=\mathbb{P}[\tau_l\le t]$ satisfies:
$$1-f(s)=\frac{1-f(s+t)}{1-f(t)}.$$
Differentiating with respect to $t$ and setting $t=0$ leads to
\begin{align*}
-f'(s)=\frac{-f'(s+t)}{1-f(t)},\\
-f'(0)=\frac{-f'(t)}{1-f(t)}.
\end{align*}
Noticing that $-f'(0)$ is a constant, the integrals of the two sides turn to be
\begin{align*}
-f'(0)t=\ln{(1-f(t))}+C\\
f(t)=1-e^{-f'(0)t},
\end{align*}
where the constant $C$ is supposed to be exactly $1$ to satisfy $f(\infty)=1$. Thus the holding time $\tau_l$ follows an exponential distribution $\mathbb{P}[\tau_l\ge t]=e^{-\lambda_l t}$ with $\lambda_l=f'(0)$ determined by $X_l$. As $\epsilon$ is given, without loss of generality, $T$ is supposed to be $2^m$ with $m=\lceil\log{(-\frac{1}{\lambda_l}\ln{\epsilon})}\rceil$ so that 
\begin{equation}
T=-\frac{1}{\lambda_l}\ln{\tilde{\epsilon}}\ge-\frac{1}{\lambda_l}\ln{\epsilon},
\end{equation}
where
\begin{equation}\label{eq:12}
\tilde{\epsilon} = e^{-\lambda_lT} < \epsilon
\end{equation}
is a more strictly error bound which is easier to compute.
On the truncated support set $[0, T)$, simple computation shows that
\begin{equation}\label{eq:13}
\bar{\mathbb{P}}[t\le\tau<t+1] = \frac{\mathbb{P}[t\le\tau<t+1]}{\mathbb{P}[0\le\tau<T]} = \frac{e^{-\lambda_l t} - e^{-\lambda_l (t+1)}}{1-e^{-\lambda_l T}}.
\end{equation}
Assuming that the preparation circuit consists of $m=\log{T}$ qubits initialized to be $\ket{0}$, the $j^{th}$ qubit is then aligned with the rotation gate $R_Y(\theta_j)$ with $\theta_j = \arctan{\tilde{\epsilon}^{1/2^{j+1}}}:1\le j \le m$. A direct computation shows that
\begin{align}
\notag&\bigotimes_{j=1}^m R_Y(\theta_j)\ket{0} \\
\notag=&\bigotimes_{j=1}^m (\frac{1}{\sqrt{1+\tilde{\epsilon}^{1/2^j}}}\ket{0} + \frac{\tilde{\epsilon}^{1/2^{j+1}}}{\sqrt{1+\tilde{\epsilon}^{1/2^j}}}\ket{1})\\
\notag=&\bigotimes_{j=1}^m(\sum_{k_j=0}^1 \frac{\tilde{\epsilon}^{k_j/2^{j+1}}}{\sqrt{1+\tilde{\epsilon}^{1/2^j}}}\ket{k_j})\\
\notag=&\sum_{k_1,k_2,...,k_m=0}^1 \frac{\tilde{\epsilon}^{\sum_{j=1}^m k_j/2^{j+1}}}{\prod_{j=1}^m\sqrt{1+\tilde{\epsilon}^{1/2^j}}}\ket{k_1k_2...k_m}\\
\notag=&\sum_{k_1,k_2,...,k_m=0}^1 \sqrt{\frac{1-\tilde{\epsilon}^{1/2^m}}{1-\tilde{\epsilon}}}\tilde{\epsilon}^{\sum_{j=1}^m \frac{k_j}{2^{j+1}}}\ket{k_1k_2...k_m}\\
=\label{eq:14}&\sum_{k_1,k_2,...,k_m=0}^1 \sqrt{\frac{\tilde{\epsilon}^{\sum_{j=1}^m \frac{k_j}{2^j}}-\tilde{\epsilon}^{\frac{1}{2^m}+\sum_{j=1}^m \frac{k_j}{2^j}}}{1-\tilde{\epsilon}}}\ket{k_1k_2...k_m}\\
\label{eq:15}=&\sum_{t=0}^{T-1} \sqrt{\frac{\tilde{\epsilon}^{\frac{t}{T}}-\tilde{\epsilon}^{\frac{t+1}{T}}}{1-\tilde{\epsilon}}}\ket{t}\\
=\label{eq:16}&\sum_{t=0}^{T-1} \sqrt{\frac{e^{-\lambda_l t} - e^{-\lambda_l (t+1)}}{1-e^{-\lambda_l T}}}\ket{t}\\
\label{eq:17}=&\sum_{t=0}^{T-1} \sqrt{\bar{\mathbb{P}}[t\le\tau<t+1]}\ket{t}
\end{align}
Here the binary number and the corresponding summation $(\overline{k_1k_2...k_m})_2 = \sum_{j=1}^m k_j2^{m-j}$ in Eq.~\eqref{eq:14} are both substituted by the decimal number $t$ varying from $0$ to $T=2^m$ in Eq.~\eqref{eq:15}. Then Eq.~\eqref{eq:16} and Eq.~\eqref{eq:17} are derived by Eq.~\eqref{eq:12} and Eq.~\eqref{eq:13}, respectively. Hence, with $m=\lceil\log{(-\frac{1}{\lambda_l}\ln{\epsilon})}\rceil$ qubits and the same number of rotation gates, the desired truncated state is successfully prepared((as illustrated in Figure  \ref{exponential_distribution_circuit} and Figure  \ref{exponential_distribution_simulation})). 
\end{proof}

\subsection{Proof of Theorem \ref{thm:2}\label{pf:2}}
\begin{proof}
Given a time interval $T$, it can be uniformly divided into $n$ pieces $[T_{j-1},T_j)$ with $T_j = \frac{j}{n}T(0\le j \le n)$ and $\tau_j=\frac{T}{n}:1\le j\le n$, and the corresponding discrete random variables are denoted by $X_j=X(T_j)$. Then the increments $Y_1 = X_1$ and $Y_j = X_j - X_{j-1}:2\le j\le n$ are independent and identically distributed random variables. The underlying distribution $\mathcal{F}$ can be prepared by Grover's method $\ket{0}_s \rightarrow \ket{\mathcal{F}}_j=\sum^{S}_{k=1}p_j(k)\ket{k}$, where $s=\log{S}$ is the number of qubits of the approximation $\tilde{\mathcal{F}}_j$ and $p_j(k)$ is the probability amplitude that the random variable $X_j$ lies in the $k^{th}$ interval of $\tilde{\mathcal{F}}$ with $\sum^{S}_{k=1}p^2_j(k)=1$ (see \cite{grover2002creating} for reference). The circuit depth of the preparation of $\tilde{\mathcal{F}}$ is at most $s$ utilizing $S$ control-rotation gates. Then the increments sequence $Y_j:1\le j\le n$ can be derived by repeating this procedure on $n$ parallel subcircuits: $\ket{Y_1}\ket{Y_2}...\ket{Y_n}=\bigotimes^n_{j=1}\ket{\mathcal{F}}_j$. The desired path variables $\{X_j = \sum_{j'=1}^i Y_{j'}: 1\le j\le n\}$ are derived by a sequence of recursive add operators on these $n$ registers as $\ket{Y_1}\ket{Y_2}...\ket{Y_n} \rightarrow \ket{Y_1}\ket{Y_1 \oplus Y_2}\ket{Y_3}...\ket{Y_n} \rightarrow ... \rightarrow \ket{Y_1}\ket{Y_1\oplus Y_2}...\ket{Y_1\oplus Y_2\oplus...\oplus Y_n} = \ket{X_1,X_2,...,X_n}$ (See Figure  \ref{levy_picture} \textbf{(c)}).  Each add operator takes $2\log{(nS)-1}$ Toffoli gates and $5\log{(nS)}-3$ C-NOT gates. A simple computation shows that the circuit depth is $S+(n-1)(2\log{(nS)+4})=O(S+n\log{(nS)})$ and the gate complexity is at most $nS+(n-1)(5\log{nS}-3+2\log{nS}-1) = O(n(S+\log{(nS)}))$ gates.
\end{proof}

\subsection{Proof of Corollary \ref{thm:3} and Corollary \ref{thm:4}\label{pf:3}}
\begin{proof}
By definition of Poisson Point Process, the distinct increments $Y_j:1\le j \le n$ are constant $1$, and the $X_j$ varies exactly from $1$ to $n$ so that the adder subcircuit introduced in Theorem \ref{thm:2} can be omitted. As a consequent, each $X_j$ can be prepared parallel on at most $\log{n}$ qubits via $\log{n}$ rotation gates. And meanwhile the holding times $\tau_j:1\le j \le n$ are assumed to follow an exponential distribution $\mathcal{F}(\lambda)$, and can thus be prepared parallel, by Theorem\ref{thm:1}, on $\frac{-\ln{\epsilon}}{\lambda}$ qubits via $\frac{-\ln{\epsilon}}{\lambda}$ rotation gates with circuit depth $1$. In summary, the qubit number, circuit depth and gate complexity are respectively $n\log{n} + n\frac{-\ln{\epsilon}}{\lambda} = n\lceil\log{(\frac{-n\ln{\epsilon}}{\lambda})}\rceil$, $\max{\{\lceil\log{n}\rceil, 1\}}= \lceil\log{n}\rceil$, and $n\log{\frac{-\ln{\epsilon}}{\lambda}} + \sum_{j=1}^n\log{j}=O(n\lceil\log{(\frac{-n\ln{\epsilon}}{\lambda})}\rceil)$.
As for the case of Compound Poisson Process, the only difference is the I.I.D. state space of jumps. Another $n\lceil\log{S}\rceil$ qubits are introduced to record these jumps, and hence the qubits number turn to be $n\lceil\log{(\frac{-nS\ln{\epsilon}}{\lambda})}\rceil$. Instead of copies of Hadamard gates in \textit{Poisson point process}, a sequence of $n-1$ adder operators are needed, asking for $(n-1)(2\lceil\log{(nS)}\rceil-1)$ Toffoli gates and $(n-1)(5\lceil\log{(nS)}\rceil-3)$ C-NOT gates. Thus the total gate complexity is $O(n\lceil\log{(-\frac{nS\ln{\epsilon}}{\lambda})}\rceil)$, and the circuit depth is $S+n\lceil\log{(nS)}\rceil$.
\end{proof}

\subsection{Proof of Theorem \ref{thm:6}\label{pf:5}}
\begin{proof}
Assuming that the $n$ pieces of the \textit{continuous Markov process} are $(X_j, \tau_j)$, and the $j^{th}$ transition time are $T_j = \sum_{j'=1}^j\tau_{j'}$, a direct computation shows that
\begin{equation}
\begin{split}
&\mathbb{P}[X_{j+1}=k_{j+1}|\land_{j'=1}^j X_{j'}=k_{j'}]\\
=&\mathbb{P}[X(T_j)=k_{j+1}|\land_{j'=1}^j X(t)=k_{j'}:T_{j'-1}\le t < T_{j'}]\\
=&\mathbb{P}[X(T_j)=k_{j+1}|X(t)=k_{j}:T_{j-1}\le t < T_{j}]\\
=&\mathbb{P}[X_{j+1}=k_{j+1}|X_j=k_j].
\end{split}
\end{equation}
This reveals that $X_{j+1}$ is totally determined by the previous state $X_{j}$ with the transition probabilities $P_{kl}=P(X_{j+1}=k|X_j=l)$, and thus the embedded sequence $\{X_j:1\le j\le n\}$ is a discrete Markov chain as claimed above. By the Markov condition again, the holding time $\tau_j$ satisfies
\begin{equation}\label{eq:18}
\mathbb{P}[\tau_j> t+s|\tau_j> t]=\mathbb{P}[\tau_j>s].
\end{equation}
Noticing that Eq.~\eqref{eq:18} is exactly the memory-less condition in Theorem\ref{thm:1}, and therefore $\tau_j$ follows an exponential distribution with the $\lambda_j$ determined by the current state $X_j$. As a consequence of the embedded discrete Markov chain and the exponential distribution holding time, a quantum \textit{continuous Markov process} state can be prepared as follows: $n\lceil\log{S}\rceil$ qubits are introduced for the storage of the discrete Markov chain, and another $n\lceil\log{(-\frac{1}{\lambda_{min}}\ln{\epsilon})}\rceil$ qubits are introduced for the holding time $\tau_j$. The initial state $\ket{X_1}=\sum_{k=1}^Sp(k)\ket{k}$ is derived by $S$ control-rotation gates via Grover's state preparation method. Following that, $n-1$ successive transition matrix operators, each of which consists of $S^2$ multi-controlled rotation gates, are applied to generate $\ket{X_j}:2\le j\le n$. The $n$ holding times can be prepared parallel, and each $\tau_j$ needs at most $s$ different exponential distribution with $S\lceil\log{(-\frac{1}{\lambda_{min}}\ln{\epsilon})}\rceil$ gates and $S$ circuit depth. In summary, the total number of gates is $S+(n-1)*S^2+n*S\lceil\log{(-\frac{1}{\lambda_{min}}\ln{\epsilon})}\rceil=O(nS\lceil S+\log{(-\frac{\ln{\epsilon}}{\lambda_{min}})}\rceil)$, and the circuit depth is $S+(n-1)*S^2+n*S=O(nS^2)$.
\end{proof}

\subsection{Proof of Theorem \ref{thm:7}\label{pf:6}}
\begin{proof}
As shown in Figure  \ref{cox_picture}, the preparation can be divided into two steps: Firstly, the stochastic process of intensity $\lambda(t)$ is prepared on parallel $n*q_\mathcal{F}$ qubits, and the circuit depth is $O(d_\mathcal{F})$ as given. Secondly, the holding time $\tau_j$ is prepared by a sequence of controlled-$V$ operators. Given a piece of Cox process and the corresponding fixed intensity $\lambda$, the holding time follows an exponential distribution as described in Theorem \ref{thm:1} so that it can be prepared within 1 circuit depth by rotation angles $\theta_j = \tilde{\epsilon}^{1/2^{j+1}}: 1\le j \le m$, where $\tilde{\epsilon}=e^{-\lambda T}$. As for varying digital-encoded $\lambda=\ket{l_1,l_2,...,l_{q_{\mathcal{F}}}}$, it can be divided into $2_{q_\mathcal{F}}$ sequences of controlled-rotation gates with $2_{q_\mathcal{F}}$ circuit depth. Hence the total circuit depth for preparing $\tau_j$ is $d_\mathcal{F}+2_{q_\mathcal{F}}$. On the other hand, the preparation of increments $Y_j$ can be implemented parallel with $d_\mathcal{G}$ circuit depth. Thus the circuit depth of the preparation is $\max{\{d_\mathcal{F}+2_{q_\mathcal{F}, d_\mathcal{G}\}}}$. To derive $X_j$ or $T_j$, $n$ adders are needed, with $nd_\mathcal{G}$ or $-\frac{n\ln{\epsilon}}{\lambda_{min}}$, respectively.
\end{proof}

\section{Proofs for Information Extraction Problem}\label{sec:8}
\subsection{Proof of Theorem \ref{thm:8}\label{pf:7}}
\begin{proof}
A direct computation shows that:
\begin{equation}
I(X,T)=\int_{t=0}^TX(t)\,dt=\sum_{j=0}^nX_j\tau_j.\label{eq:1}
\end{equation}
The expected value of Eq.~\eqref{eq:1} is evaluated as:
\begin{align}
\mathbb{E}[f(I(X,T))]=&\sum_{X(t)\in \bar{\Omega}_n}f(\int_{t=0}^TX(t)\,dt)\mathbb{P}[X(t)]\notag\\
=&\sum_{X(t)\in \bar{\Omega}_n}f(\sum_{j=0}^nX_j\tau_j)\mathbb{P}[X(t)]\notag\\
=&\sum_{X(t)\in \bar{\Omega}_n}f(\sum_{j=0}^nZ_j)\mathbb{P}[X(t)],\label{eq:3}
\end{align}
where $Z_j=X_j\tau_j$ can be derived via a quantum multiplication circuit introduced in Lemma \ref{thm:10} with complexity $O(l_xl_\tau)$. The total circuit depth is $O(\mathcal{P}+nl_xl_\tau)$ (as shown in Figure  \ref{extraction} \textbf{(c)}).
Noticing that the variables of $f$ in Eq.~\eqref{eq:3} can be evaluate through a standard procedure (one can see \cite{montanaro2015quantum} for reference). To guarantee the theoretical closure, we shall sketch this procedure as follows. Suppose that the function $f$ can be truncated on an interval of length $P$, then it can be extended to a $P-periodic$ function $f_P$. The Fourier approximation of order $L$ for $f_P$ is denoted by $f_{P,L}(x)=\sum_{l=-L}^L c_l e^{i\frac{2\pi l}{p}x)}$, then Eq.~\eqref{eq:3} is evaluated as
\begin{align}
&\mathbb{E}[f_{P,L}(I(X,T))]=\mathbb{E}[f_{P,L}(\sum_{j=0}^nZ_j)]\notag\\
=&\sum_{l=-L}^L c_l \mathbb{E}[e^{i\frac{2\pi l}{p}(\sum_{j=1}^nZ_j))}]\notag\\
=&\sum_{l=-L}^L c_l (\mathbb{E}[\cos{(\frac{2\pi l}{p}(\sum_{j=1}^nZ_j))}]+i\mathbb{E}[\sin{(\frac{2\pi l}{p}(\sum_{j=1}^nZ_j))}]),\label{eq:5}
\end{align}
Each term $\mathbb{E}[\cos{(\frac{2\pi l}{p}(\sum_{j=0}^nZ_j))}]$ and $\mathbb{E}[\sin{(\frac{2\pi l}{p}(\sum_{j=1}^nZ_j))}]$ in Eq.~\eqref{eq:5} can be evaluated via the standard amplitude estimation algorithm. Given the error bound $\epsilon$, the circuit should be repeated $O(1/\epsilon)$ times for each term. The total time complexity is hence $O(\frac{LP}{\epsilon}(\mathcal{P}+nl_xl_\tau))$
\end{proof}

\subsection{Proof of Theorem \ref{thm:9}\label{pf:8}}
\begin{proof}
A direct computation shows that:
\begin{align*}
J(X,T)=&\int_{t=0}^Tg(t)X(t)\,dt)\notag\\
=&\sum_{j=1}^n X_jG_j\notag\\
=&\sum_{j=1}^n(\sum_{j'=1}^j Y_{j'})G_j\notag\\
\end{align*}
where $G_j=\int_{T_{j-1}}^{T}g(t)\,dt$ is a piece-wise integral on the interval $[T_{j-1}, T]$ with $T_j = \sum_{j'=1}^j\tau_{j'}$ and $T_0=0$, and $Y_{j'}$ is the $j'^{th}$ increment. By exchanging the order of summation, one has that:
\begin{align}
J(X,T)=&\sum_{j'=1}^nY_{j'}(\sum_{j=j'}^nG_j)\notag\\
\label{eq:8}
=&\sum_{j'=1}^nY_{j'}F(T-T_{j'-1})\\
\label{eq:2}
=&\sum_{j'=1}^nY_{j'}G(T_{j'-1}),
\end{align}
where $F(T')=\int_{T-T'}^T g(t)\,dt$ in Eq.~\eqref{eq:8} satisfies:
\begin{equation*}
F(T-T_{j'-1})=\int_{T_{j'-1}}^T g(t)\,dt=\sum_{j=j'}^nG_j,\notag\\
\end{equation*}
and $G(T')=\int_{T'}^T g(t)\,dt$ in Eq.~\eqref{eq:2} satisfies:
\begin{equation*}
G(T_{j'-1})=\int_{T_{j'-1}}^T g(t)\,dt=\sum_{j=j'}^nG_j.\\
\end{equation*}
Thus the expected value of Eq.~\eqref{eq:8} and Eq.~\eqref{eq:2} can be evaluated as
\begin{align}
\mathbb{E}[f(J(X,T)]=&\sum_{X(t)\in \bar{\Omega}_n}f(\sum_{j=1}^nY_jF(T-T_{j'-1}))\mathbb{P}[X(t)]\label{eq:9}\\
=&\sum_{X(t)\in \bar{\Omega}_n}f(\sum_{j=1}^nY_jG(T_{j'-1}))\mathbb{P}[X(t)]\label{eq:10}\\
=&\sum_{X(t)\in \bar{\Omega}_n}f(\sum_{j=1}^nW_j)\mathbb{P}[X(t)].\label{eq:4}
\end{align}
The state $F(T-T_{j'-1})$ in Eq.~\eqref{eq:9} is a function of the variable $T-T_{j'-1}$, and hence can be computed via a operator $F$ on the $T-T_{j'-1}$ qubit as illustrated in Figure  \ref{extraction} \textbf{(d)}.  And the directed area $W_j=Y_jF(T-T_{j'-1})$ can be derived through a quantum multiplier introduced in Lemma \ref{thm:10} with complexity $O(l_xl_\tau)$. In the same way the state $G(T_{j'-1})$ in Eq.~\eqref{eq:10} is a function of the variable $T_{j'-1}$, and hence can be computed via a operator $G$ on the $T_{j'-1}$ qubit. And the directed area $W_j=Y_jG(T_{j'-1})$ can be derived with complexity $O(l_Yl_T)$ (see Figure  \ref{extraction} \textbf{(e)} for reference). The total circuit depth is $O(\mathcal{P}+\mathcal{F}+nl_xl_\tau+n\max{(l_x,l_\tau)})$ or $O(\mathcal{P}+\mathcal{G}+nl_Yl_T)$.
Noticing that the expression in Eq.~\eqref{eq:4} once again take the form of a summation, and thus can be solved through a standard procedure of amplitude estimation as mentioned above: The Fourier approximation of order $L$ for $f_P$ is denoted by $f_{P,L}(x)=\sum_{l=-L}^L c_l e^{i\frac{2\pi l}{p}x)}$, then Eq.~\eqref{eq:4} is evaluated as
\begin{align}
&\mathbb{E}[f_{P,L}(J(X,T))]=\mathbb{E}[f_{P,L}(\sum_{j=1}^nW_j)]\notag\\
=&\sum_{l=-L}^L c_l \mathbb{E}[e^{i\frac{2\pi l}{p}(\sum_{j=1}^nW_j))}]\notag\\
=&\sum_{l=-L}^L c_l (\mathbb{E}[\cos{(\frac{2\pi l}{p}(\sum_{j=1}^nW_j))}]+i\mathbb{E}[\sin{(\frac{2\pi l}{p}(\sum_{j=1}^nW_j))}]).\label{eq:6}
\end{align} 
Each term $\mathbb{E}[\cos{(\frac{2\pi l}{p}(\sum_{j=1}^nW_j))}]$ and $\mathbb{E}[\sin{(\frac{2\pi l}{p}(\sum_{j=1}^nW_j))}]$ in Eq.~\eqref{eq:6} can be evaluated via rotation operators  on the target qubit controlled by $W_j$ together with a standard amplitude estimation subcircuit.(as shown in the box in Figure  \ref{extraction} \textbf{(d)} and Figure  \ref{extraction} \textbf{(e)}. Given the error bound $\epsilon$, the circuit should be repeated $O(1/\epsilon)$ times for each term. The total time complexity is hence $O(\frac{LP}{\epsilon}(\mathcal{P}+\mathcal{F}+nl_xl_\tau+n\max{(l_x,l_\tau)}))$ or $O(\frac{LP}{\epsilon}(\mathcal{P}+\mathcal{G}+nl_Yl_T))$ for holding time and increment representation, respectively.
\end{proof}

\end{document}